%% file: Main.tex
\def\IEEEsubmission{0}
\if\IEEEsubmission1
\documentclass[journal,12pt,onecolumn,draftclsnofoot]{IEEEtran}
\else
\documentclass[journal]{IEEEtran}
\fi
\usepackage{multicol}
\usepackage{mathtools}
\usepackage{amsthm}
\newtheorem{theorem}{Theorem}
\newtheorem{lemma}[theorem]{Lemma}
\usepackage{amsfonts}
\usepackage[bookmarksopen=true]{hyperref}
\usepackage{stfloats}
\usepackage{acro}
\usepackage[noadjust]{cite}
\usepackage{multirow}
\usepackage{bm}
\usepackage{enumitem}
\usepackage{amsmath,amssymb}
\usepackage{graphicx}
\usepackage{epstopdf}
\usepackage[table,xcdraw]{xcolor}
\usepackage[geometry]{ifsym}
\usepackage{array}
\usepackage[utf8]{inputenc}
\usepackage[T1]{fontenc}
\usepackage{pifont}

\usepackage{algorithmic}
\usepackage{algorithm}
\usepackage[caption=false,font=normalsize,labelfont=sf,textfont=sf]{subfig}
\usepackage{textcomp}
\usepackage{url}
\usepackage{verbatim}
\usepackage[normalem]{ulem}

\def\BibTeX{{\rm B\kern-.05em{\sc i\kern-.025em b}\kern-.08em
		T\kern-.1667em\lower.7ex\hbox{E}\kern-.125emX}}

\input{Acronyms}
\usepackage{cite}

\begin{document}

\title{
Interference Mitigation and Spectral Efficiency Enhancement in a Multi-BD Symbiotic Radio
}

\author{Fikiri Salum Uledi, Muhammad Bilal Janjua, \c{C}a\u{g}r{\i} \"{O}zgen\c{c} Etemo\u{g}lu and H\"{u}seyin Arslan, ~\IEEEmembership {Fellow,~IEEE}
\thanks{F. S. Uledi is with the Department of Electrical and Electronics Engineering, Istanbul Medipol University, Istanbul, 34810, T\"{u}rkiye, and also with the Department of Electronic and Telecommunications, University of Dar es Salaam, Tanzania (email: fikiri.uledi@std.medipol.edu.tr). H. Arslan is with the Department of Electrical and Electronics Engineering, Istanbul Medipol University, Istanbul, 34810, T\"{u}rkiye (email: huseyinarslan@medipol.edu.tr). M. B. Janjua is with the R\&D Department, Oredata, Istanbul, T\"{u}rkiye (email: bilal.janjua@oredata.com). \c{C}. \"{O}. Etemo\u{g}lu is with the T\"{u}rk Telekom R\&D Department, Istanbul, T\"{u}rkiye (email: cagriozgenc.etemoglu@turktelekom.com.tr).}  
}



\maketitle
\begin{abstract}
This study presents a framework designed to mitigate direct-link interference (DLI) and inter-backscatter device interference (IBDI) in multi-backscatter orthogonal frequency division multiplexing (OFDM)-based symbiotic radio (SR) systems. The framework employs OFDM signal designs with strategic allocation of null subcarriers and incorporates two backscatter modulation techniques: on-off frequency shift keying (OFSK) and multiple frequency shift keying (MFSK) for symbiotic backscatter communication (SBC). Additionally, we propose Fully-Orthogonal and Semi-Orthogonal multiple access schemes to facilitate SBC alongside primary communication. The Fully-Orthogonal scheme maintains orthogonality between direct link and SBC signals, thereby ensuring interference-free SBC, albeit at a reduced spectral efficiency. In contrast, the Semi-Orthogonal schemes eliminate IBDI but permit partial DLI, striking a balance between reliability and spectral efficiency. To address the partial DLI inherent in Semi-Orthogonal schemes, successive interference cancellation (SIC) is employed at the receiver, enhancing SBC reliability. To tackle channel estimation challenges in SBC within the SR system, we implement non-coherent detection techniques at the receiver. The performance of the proposed system is evaluated based on average bit error rate (BER) and sum-rate metrics, demonstrating the effectiveness of our schemes. We provide analytical results for the system's detection performance under both proposed modulation techniques and multiple access schemes, which are subsequently validated through extensive simulations. These simulations indicate a notable error-rate reduction of up to $10^{-3}$ at $20$ dB with the Fully-Orthogonal scheme with MFSK.
\end{abstract}
\vspace{-1 em}
\section{Introduction}
\IEEEPARstart{T}{he} rapid evolution of \ac{IoT} technologies has led to an unprecedented increase in the number of connected devices, with projections estimating 30 billion devices by 2027, as highlighted in the Ericsson Mobility Report \cite{ref1}. This surge presents significant challenges in spectrum allocation, energy consumption, and network scalability in wireless communication networks. To address these issues, significant efforts are being made to develop spectrally and energy-efficient communication techniques for next-generation sustainable \ac{IoT} systems, which either operate without batteries or harvest energy from ambient sources \cite{ref2,ref3,ref4}. In this pursuit, \ac{3GPP}, through \ac{5G-NR} standards in Release 18 and Release 19, introduces passive \ac{IoT} or ambient power-enabled \ac{IoT}, which focuses on ultra-low-power communication technologies such as ambient \ac{BC} \cite{ref5}. Unlike conventional \ac{BC}, which is realized through \acp{RFID} that rely on dedicated \ac{RF} exciters\cite{ref6}, ambient \ac{BC} eliminates this dependency by leveraging ambient signals from TV towers, cellular \ac{BS}, Wi-Fi access points, and AM/FM transmitters as carrier signals on which a \ac{BD} passively modulates its information signal \cite{ref7, ref8, ref9}. However, in ambient \ac{BC}, the signal transmitted by \ac{BD} faces \ac{DLI} from the primary signal and \ac{IBDI} from other \acp{BD} signals at the receiver, either unintentionally or in direct competition for \ac{RF} resources \cite{ref10}. The impact of \ac{IBDI} becomes more severe in dense \acp{BD} deployment because uncoordinated transmissions can cause collisions, resource contention, and overall system performance degradation. 

Since most modern wireless systems, including cellular and Wi-Fi networks, rely on \ac{OFDM} as the standard waveform, this presents an opportunity for the \ac{BC} to utilize \ac{OFDM} as the signal source at the \ac{BS} while leveraging its features to address the \ac{DLI} and \ac{IBDI} problems. Several approaches are available in the literature that exploit \ac{OFDM} features in different domains to solve the interference problem in ambient \ac{BC}. In \cite{ref15}, authors apply the subcarrier-wise \ac{BC} concept to transmit one information bit per subcarrier cluster, applying a notch filter bank to attenuate signals over clusters dedicated to \ac{BC} while leaving the rest unchanged. The requirement for a notch filter increases both hardware and computational complexity at the \ac{BD}. In \cite{ref16,ref17,ref18}, authors propose to use the uncorrupted portion of the \ac{CP} to mitigate \ac{DLI}. However, this approach requires the receiver to have an accurate knowledge of the maximum channel delay, which may not always be feasible in practical scenarios. In \cite{ref19,ref20,ref21}, the authors utilize preambles and pilots in \ac{OFDM} for \ac{BC}. However, repurposing preambles and pilots for \ac{BC} may compromise synchronization and channel estimation of the primary signal transmission link, especially in dynamic environments. In \cite{ref22,ref23,ref24} the authors apply the \ac{FSK} modulation technique to modulate the \ac{BC} signal in the edge subcarriers and guard bands, but this approach necessitates an additional filter to scan the adjacent channels, leading to increased hardware and computational complexity. While ambient \ac{BC} solutions utilizing \ac{OFDM} systems offer certain gains, their reliability is constrained by interference resulting from the lack of cooperation \cite{ref11, ref12}. Moreover, the challenge in ambient \ac{BC} systems lies not only in managing interactions between the \acp{BD} and the transmission of ambient signals but also in mitigating destructive competition among \acp{BD} for limited resources.

To overcome the limitations of ambient \ac{BC}, a new concept of \ac{SR} has been introduced, where the primary system (i.e., existing Wi-Fi, cellular, or any other wireless system) not only serves its users but also supports \ac{BC} by sharing its RF resources such as spectrum and energy by involving more coordinated approach. Unlike traditional ambient \ac{BC}, where \acp{BD} access the signals of the primary system blindly, \ac{SR} fosters a cooperative relationship by transforming \acp{BD} from strange entities into cooperative network components. The communication behavior of the \acp{BD} in SR varies based on their interaction with primary signal transmission and other \acp{BD}' transmissions, ranging from mutualistic, where \acp{BD} enhance \ac{Rx} reception with spatial and multi-user diversity gain with minimal interference, to commensal, where \acp{BD} benefit from the primary signal transmission system without any impact \cite{ref13, ref14}. Authors in \cite{ref25} propose \ac{OOK}, \ac{FSK}-1, and \ac{FSK}-2 schemes, in which \ac{BD} modulates information signal by shifting the incoming \ac{OFDM} signal onto pre-allocated null subcarriers for \ac{SR-BC}. However, this approach is limited to a single \ac{BD} and exhibits spectral inefficiency, particularly in the cases of \ac{FSK}-1 and \ac{FSK}-2. In \cite{ref26}, two multiple access schemes are proposed to enhance massive IoT connectivity in SBC, simultaneous access (SA), which allows multiple IoT devices to transmit simultaneously by backscattering the \ac{BS} signal, and selection diversity access (SDA), which selects the device with the strongest backscatter link to reduce interference. However, the SA scheme faces increased interference and signal degradation, whereas the SDA scheme may lead to inefficient resource utilization.  To this extent, the existing literature reveals a knowledge gap in addressing both \ac{DLI} and \ac{IBDI} in multi-\ac{BD} \ac{SR-BC} systems. Most solutions focus primarily on mitigating \ac{DLI}, particularly in single-\ac{BD} scenarios, with limited attention to \ac{IBDI}.

\subsection{Key Contributions}
In this work, we propose schemes for interference mitigation and spectral efficiency enhancements in a multi-\ac{BD} \ac{OFDM}-based SR system. The main contributions of this paper are summarized as follows:

\begin{itemize}
\item 
We propose \ac{OFDM} signal design based on a Fully-Orthogonal scheme for multiple access of \acp{BD} in the \ac{SR} system. Specifically, null subcarriers are left between the data subcarriers throughout the entire \ac{OFDM} symbol for orthogonal reception of \acp{BD} and direct link signals. Although this approach effectively mitigates the \ac{DLI} and \ac{IBDI}, reserving subcarriers for \ac{SR-BC} reduces the number of subcarriers available for the primary data transmission. Consequently, while interference-free \ac{SR-BC} is achieved, it comes at the cost of the spectral efficiency of \ac{SR} system.

\item  
To improve the spectral efficiency of the system, we introduce a Semi-Orthogonal scheme for multiple access of \acp{BD}, where a limited number of null subcarriers are allocated for \ac{SR-BC}. This scheme provides more subcarriers for primary data transmission while intentionally allowing manageable \ac{DLI}. Taking advantage of inherent \ac{SR} resource sharing, we apply the \ac{SIC} at the receiver to cancel the \ac{DLI} from the \ac{BD} signal. Thus, efficiently exploiting \ac{BD} signal in the Semi-Orthogonal scheme.  

\item We introduce \ac{OFSK} and \ac{MFSK} modulation techniques at the \acp{BD} for modulating their data over the primary signal. For information transfer, \acp{BD} employs these modulation techniques to shift the data in the primary signal to subcarriers designated for the multiple access of the \acp{BD}. To avoid the complexities of channel estimation, a non-coherent detector is implemented at the receiver to detect the \acp{BD} information from the received signal. 

\item 
We demonstrate the efficacy of a non-coherent detector for both modulation techniques. We analyze the achievable sum-rate and error rates of \ac{SR} analytically and through comprehensive simulations. This analysis shows the reliability and spectral efficiency of the proposed schemes in \ac{SR} system. Additionally, we analyze the performance of our proposed schemes in the presence of the \ac{CFO}. To alleviate \ac{CFO} impairment, we introduce a compensation method aimed at minimizing detection errors and enhancing non-coherent detector performance.

\end{itemize}

{\em Organization:} The rest of the paper is organized as follows: Section II provides explanations of the design of the \ac{SR} system along with the foundational concepts used throughout the subsequent sections. Section III presents the proposed modulation techniques. Section IV presents detection procedures and performance analysis. In Section V, we present the simulation results and discuss the theoretical findings. Conclusions and observations are provided in Section VI.

{\em Notation:} $\mathbb{Z}^+$ represents the set of positive integers, $\mathcal{CN}(0, \sigma^2)$  denotes the circularly symmetric complex Gaussian distribution with zero mean and variance $\sigma^2$, while $\mathbb{E}[\cdot]$ represents the expectation of its argument over random variables. The probability of an event $A$ is denoted by $\Pr(A)$, while the conditional probability of $A$ given $B$ is expressed as $\Pr(A \mid B)$. The floor operator, which returns the greatest integer less than or equal to a given number, is denoted by $\lfloor \cdot \rfloor$. Additionally, the operator $\ast$ represents the linear convolution operation.

\section{Multi-BD OFDM-based SR System Model}
We consider a multi-\ac{BD} \ac{OFDM}-based \ac{SR} system in \figurename~\ref{fig:image1}, where a $P$ number of single-antenna \acp{BD} communicate to \ac{Rx} through the backscatter links denoted by $h_{\mathrm{b}1},h_{\mathrm{b}2}....,h_{\mathrm{b}P}$ by passively modulating their information signals into the reflections of the signal coming from the \ac{BS} through the forward links represented by $h_{\mathrm{f}1},h_{\mathrm{f}2}....,h_{\mathrm{f}P}$. Additionally, the \ac{BS} communicates directly with the \ac{Rx} through a direct link, denoted by $h_{\mathrm{d}}$, whereby all \acp{BD} are assumed to be low-power, low-cost, and low-complexity devices, characteristic of a typical passive \ac{IoT} system.

\begin{figure*}[t]
    \centering
    \includegraphics[width=0.9\linewidth]{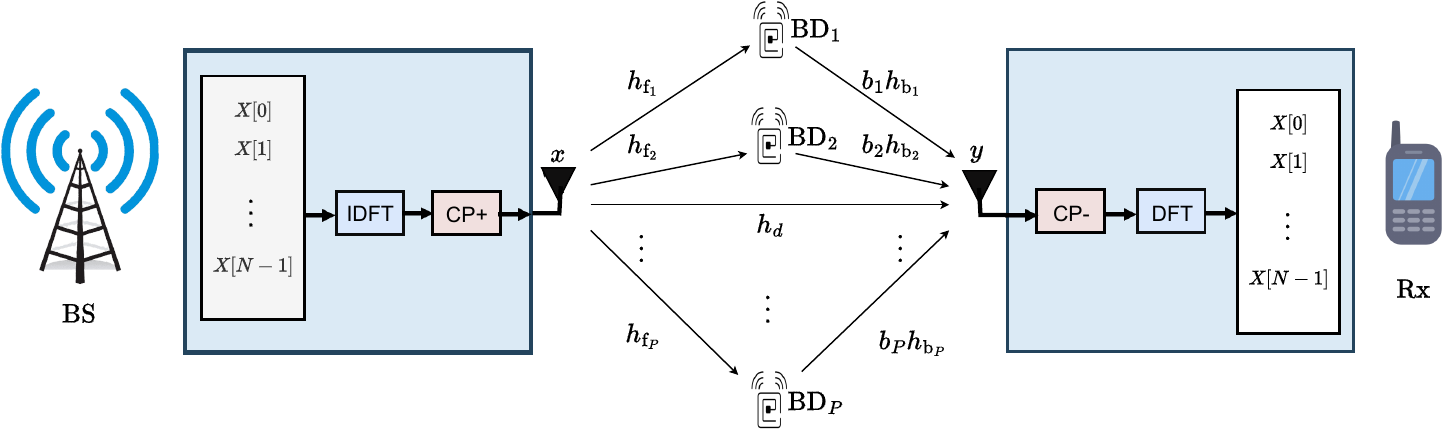}
    \caption{The proposed multi-BD OFDM-based symbiotic radio system model.}
    \label{fig:image1}
\end{figure*}

\subsection{Primary Signal}
The \ac{BS} transmits an \ac{OFDM} symbol containing $N$ subcarriers carrying the data symbols intended for the \ac{Rx}. The OFDM time-domain symbol is obtained through the Inverse Discrete Fourier Transform (IDFT), and the \textit{n}-th time-domain sample is expressed as
\begin{equation}
x[n] = \frac{1}{\sqrt{N}} \sum_{k=0}^{N-1} X[k] e^{j \frac{2\pi k n}{N}}, \quad 0 \leq n \leq N-1~,
\end{equation}
where also N is the IDFT size. After the IDFT operation, a {CP} is appended to the beginning of each time-domain symbol to generate the final time-domain \ac{OFDM} signal, which is then transmitted to the receiver after a passband up conversion. The CP length is chosen to be greater than the maximum excess delay of the channel (or delay spread) to eliminate \ac{ISI} when consecutive OFDM symbols are transmitted over the wireless channel. The passband signal transmitted by the \ac{BS} can be expressed as
\begin{equation}
\label{deqn_ex1a}
x(t) = x_{\text{cp}}(t) e^{j 2 \pi f_c t}~,
\end{equation}
 where $x_{\text{cp}}(t)$ is the baseband \ac{OFDM} signal with CP in the time domain, and $f_{c}$ denotes the carrier frequency.

\subsection{BD Signal}
\acp{BD} encode information by dynamically adjusting their load impedances in multiple discrete states, altering the impedance relationship between the antenna impedance ${Z_\mathrm{a}}_p$ and the load impedance ${Z_\mathrm{l}}_p$. These impedance variations create conditions of matching or mismatching, corresponding to the transmission of binary bits `0' and `1', respectively. When ${Z_\mathrm{a}}_p = {Z_\mathrm{l}}_p$, perfect impedance matching occurs, allowing maximum power transfer from the antenna to the load. On the other hand, when ${Z_\mathrm{a}}_p \neq {Z_\mathrm{l}}_p$, an impedance mismatch occurs, causing the maximum energy reflection from the load back to the antenna \cite{ref6} where it is exploited for \ac{BC}. The reflection coefficient $\alpha_p$, which governs how \acp{BD} encode information by controlling signal reflection and absorption, is mathematically expressed as \cite{ref27}:

\begin{equation} \label{deqn_ex1b} 
    \alpha_p = \frac{{Z_\mathrm{l}}_p - {Z_\mathrm{a}}_p^*}{{Z_\mathrm{l}}_p - {Z_\mathrm{a}}_p}= |\alpha_p| e^{j \theta_p}~. 
\end{equation}
By varying ${Z_\mathrm{l}}_p$ across different non-zero values while keeping ${Z_\mathrm{a}}_p$ constant, as dictated by the antenna structure, different values of $\alpha_p$ are generated. With, $\theta_p$ representing the phase shift that influences how the signal propagates or reflects in the communication channel. To further characterize the behavior of $\alpha_p$, the following equations are used \cite{ref28}:

\setcounter{equation}{4} 
\begin{equation} 
    |\alpha_p| = \frac{|{Z_\mathrm{a}}_p| + |{Z_\mathrm{l}}_p| - 2|{Z_\mathrm{a}}_p||{Z_\mathrm{l}}_p|\cos({\theta_\mathrm{a}}_p - {\theta_\mathrm{l}}_p)}{|{Z_\mathrm{a}}_p| + |{Z_\mathrm{l}}_p| + 2|{Z_\mathrm{a}}_p||{Z_\mathrm{l}}_p|\cos({\theta_\mathrm{a}}_p - {\theta_\mathrm{l}}_p)} \tag{\theequation a}~, 
\end{equation}

\begin{equation} 
    \theta_p = \arctan \left( \frac{2|{Z_\mathrm{a}}_p||{Z_\mathrm{l}}_p|\sin({\theta_\mathrm{a}}_p - {\theta_\mathrm{l}}_p)}{|{Z_\mathrm{a}}_p|^2 + |{Z_\mathrm{l}}_p|^2} \right) \tag{\theequation b}~, 
\end{equation} 
\setcounter{equation}{4}
where ${Z_\mathrm{a}}_p = |{Z_\mathrm{a}}_p| e^{j{\theta_\mathrm{a}}_p}$ and ${Z_\mathrm{l}}_p = |{Z_\mathrm{l}}_p| e^{j{\theta_\mathrm{l}}_p}$, with ${\theta_\mathrm{a}}_p$ and ${\theta_\mathrm{l}}_p$ representing their respective phase angles. The term ${\theta_\mathrm{a}}_p - {\theta_\mathrm{l}}_p$ denotes the phase difference between these impedances, and $\cos(\theta_a - {\theta_\mathrm{l}}_p)$ provides insight into the real part of their interaction.

To achieve a gradual transition in the reflection coefficient, \acp{BD} continuously modify the load impedance ${Z_\mathrm{l}}_p$, which can be fine-tuned by adjusting inductive elements \cite{ref9}. By alternating between multiple impedance states at different rates, results in the frequency shift of the incident \ac{OFDM} symbol from the \ac{BS} by the \acp{BD} which also accounts for the shift in subcarriers positions, as detailed in \cite{ref25}.

Traditional \ac{BD} performs modulation by switching between two discrete impedance states using rectangular pulses \cite{ref6}. However, the transitions between these states may introduce undesired frequency shifts, potentially causing the \ac{BD} to transmit in the out-of-band region of the primary communication system where it is highly probable to encounter interference from other systems' transmissions. To mitigate these undesired frequency shifts, pulse shaping techniques are employed, enabling continuous load variations through variable impedance elements controlled by a \ac{MCU}. By adjusting voltage levels applied to a diode, continuous load modulation is achieved, reducing out-of-band emissions and refining the spectral characteristics of the reflected signal. Some advanced \ac{BD} architecture further implements single-sideband \ac{FSK} to suppress higher-order harmonics and unwanted mirror frequencies \cite{ref29}.

In this context, we assume that \acp{BD} can generate complex exponential signals at specific frequencies and perform both single-sideband \ac{FSK} and \ac{OFSK} modulations to manipulate the time-frequency properties of the backscattered signal. Additionally, we assume that the \acp{MCU} in \acp{BD}, along with other essential low-power operations, are powered through ambient energy harvesting mechanisms such as RF or solar energy \cite{ref30}, \cite{ref31}. When the $\mathrm{BD}_p$ employs frequency modulation for data transmission, the passively modulated signal can be expressed as
\begin{equation} 
    x_{\mathrm{b}_p}(t) = (h_{\mathrm{f}p}(t) * x(t))\alpha_p(t)b_p(t)~, 
\end{equation}
where $b_p(t)$ represents the information signal generated at the $\mathrm{BD}_p$.

\subsection{Received Signal}
The signal received at the \ac{Rx} is the superposition of transmitted signals from the \acp{BD} and from  the \ac{BS}, which can be expressed as
\begin{equation}
y(t)= x(t) * h_{\mathrm{d}}(t) +\sum_{p=1}^P h_{\mathrm{b}p}(t) *x_{\mathrm{b}p}(t)+ w(t)~,
\end{equation}
where $ w(t) $ represents \ac{AWGN}, $ h_{\mathrm{d}}(t)$ and $h_{\mathrm{b}p}(t)$ denote the continuous-time channel responses for the direct link and the backscatter links of the $\mathrm{BD}_p$. All links exhibit multipath Rayleigh fading with maximum excess delay of the composite channel expressed as
$\tau_{\text{max}} = \max \left\{ \tau_\mathrm{d}, \tau_{\mathrm{f}p} + \tau_{\mathrm{b}p}, \dots, \tau_{fP} + \tau_{bP} \right\}$, where $\tau_\mathrm{d}$ represents the delay of the direct link, while $\tau_{\mathrm{f}p}$ and $\tau_{\mathrm{b}p}$ denote delays for forward and backscatter links, respectively. Each path follows multipath Rayleigh fading with maximum excess delay of the composite channel expressed as in the forward and backscatter links for the $\mathrm{BD}_p$, respectively, where $ p = 1, 2, \dots, P $. After down-converting the received signal and removing the \ac{CP}, the discrete-time representation of the received signal can be expressed as
\begin{equation}
y[n] = y_{\mathrm{d}}[n] + \sum_{p=1}^Py_{\mathrm{b}p}[n] + w[n]~,
\end{equation}
where $y_{\mathrm{d}}[n] = x[n] * H_{\mathrm{d}}[n] $ is the signal received from \ac{BS} via the direct link, $ y_{\mathrm{b}p}[n] = \alpha_p h_{\mathrm{b}p}[n] * ((h_{\mathrm{f}p}[n]*x[n])b_p[n]) $ is the signal received from the $\mathrm{BD}_p$ via a backscatter link $p$ and $ w[n] \sim \mathcal{CN}(0, \sigma_w^2) $ is the \ac{AWGN} with zero mean and variance $ \sigma_w^2 $.

In this co-existence scenario, both the direct link and backscatter link transmit the same signal. However, at the receiver, the direct link signal $y_\mathrm{d}$ appears significantly stronger, causing substantial interference to the backscattered signal $y_{\mathrm{b}_p}$. This power difference makes it challenging for the \ac{Rx} to reliably demodulate backscattered information. To mitigate this interference, time division multiplexing (TDM) and frequency division multiplexing (FDM) can be employed to allocate separate time slots and frequency bins between the primary signal and the \ac{BD} signal.

\section{Proposed Modulation Techniques and Multiple Access Schemes}

Since \ac{OFDM} distributes data symbols across multiple subcarriers in the frequency domain \cite{ref32}, \cite{ref33}, the impact of interference is well-suited for study at the subcarrier level. Building upon various interference mitigation techniques in wireless communications discussed in \cite{ref34}. In this work, we propose subcarrier shift-based techniques to manage interference in \ac{SR-BC}. Specifically, we propose the \ac{OFSK} and \ac{MFSK} modulation techniques at the \acp{BD}, which introduce frequency shifting to the incoming \ac{OFDM} signal, assigning the information signals from multiple \acp{BD} to subcarriers designated for \acp{BD} signal transmissions based on utilized multiple access scheme. In the Fully-Orthogonal scheme, all designated subcarriers are left null to fully accommodate the \acp{BD} signals, which significantly reduces interference and enhances the reliability of \ac{SR-BC}. However, it comes at the cost of spectral inefficiency. To enhance spectral efficiency, we apply a Semi-Orthogonal scheme, where only a few number of the designated subcarriers are left null while the remaining subcarriers continue to be used for primary data transmission. Under this scheme, the frequency shift induced by the applied modulation technique at the \ac{BD}, shifts the \ac{BD} signal to both null and non-null designated subcarriers. \acp{BD} signals encounter \ac{DLI} at non-null subcarriers while experiencing no interference at designated null subcarriers. The presence of \ac{DLI} at non-null subcarriers arises from the occupancy of the direct link signal within those subcarriers. Although this approach preserves spectral efficiency by offering a larger number of subcarriers for primary signal transmission, it introduces partial \ac{DLI} which significantly degrades detection performance and affects the reliability of \ac{SR-BC}. In this context, we reserve $N_b=\sum_{p=1}^P N_{{bp}}$ as the number of null subcarriers for the \acp{BD} signal transmissions and $N_\mathrm{d}$ as the number of data subcarriers for primary data transmissions dedicated for \ac{Rx}, where $N_{{bp}}$ represents the subcarriers allocated to $\text{BD}_p$.

\subsection{Proposed SBC Modulation Schemes}
\subsubsection{OFSK modulation}
 \ac{OFSK} is an \ac{OOK}  with frequency-shifting features, which conventionally, in \ac{SR-BC} and \ac{BC} in general, it is implemented by switching the antenna between reflecting and non-reflecting states, thereby altering the backscatter signal. When transmitting the information bit `1', \ac{BD} alternates its antenna impedances at the rate at which \ac{OOK} performs a particular frequency shift in the incoming \ac{OFDM} symbol. Conversely, when transmitting information bit `0', the \ac{BD} maintains its antenna impedance in a fixed state, resulting in no change in the frequency shift in the incoming \ac{OFDM} symbol. Thus, it is referred to as \ac{OFSK}. For this modulation technique, the waveform of the $\mathrm{BD}_p$ signal can be expressed as

\begin{equation}
b_p[n] =
\begin{cases} 
\alpha_p, & \text{if } b =  0 \\
\alpha_pe^{j 2\pi f_p n}, & \text{if } b =  1
\end{cases}~,
\end{equation}
where $f_p$ is the frequency shift corresponding to $p$ subcarrier shift. 

\subsubsection{MFSK modulation Scheme}
To enhance the reliability of multi-BD \ac{OFDM}-based \ac{SR-BC}, we introduce \ac{MFSK}, a multiple single-side multi subcarrier \ac{FSK} scheme. This technique utilizes multiple subcarriers to modulate the information signals of different \acp{BD} per pair of distinct subcarriers in such a way that to transmit the information bit `0', the $\mathrm{BD}_p$ performs $p$ subcarrier shift, conversely to transmit the information bit `1', it performs $p+1$ subcarrier shift towards the single side. For this modulation technique, the waveform of the $\mathrm{BD}_p$ signal can be expressed as
\begin{equation}
    b_p [n] =
\begin{cases} 
e^{j \frac{2\pi f_{p} n}{N}}, & \text{if } b =  0 \\
e^{j \frac{2\pi f_{p+1} n}{N}}, & \text{if } b =  1
\end{cases}~,
\end{equation}
where $f_{p+1}$ is the frequency shift corresponding to $p+1$ subcarrier shift.

\subsection{Fully-Orthogonal Schemes }
\subsubsection{Fully-Orthogonal Scheme with OFSK}
We design the \ac{OFDM} signal with $P$ null subcarriers allocated after each data subcarrier throughout the \ac{OFDM} symbol to facilitate separate signal transmission from each \ac{BD}. Each \ac{BD} applies a distinct frequency shift to the incoming \ac{OFDM} signal which introduces the subcarrier shift, in such a way that when transmitting an information bit ‘1,’ $\mathrm{BD}p$ shifts the signal by $p$ subcarriers, while $\mathrm{BD}{p+1}$ shifts it by $p+1$ subcarriers. As a result, at the \ac{Rx}, the signals from different \acp{BD} appear on distinct subcarriers. Conversely, when transmitting an information bit ‘0,’ $\mathrm{BD}p$ forwards the incoming \ac{OFDM} signal to the \ac{Rx} with an attenuation factor of $\alpha_p$, likewise for $\mathrm{BD}{p+1}$. This technique ensures orthogonality between the \ac{BS} and \acp{BD} signals at the \ac{Rx}, as illustrated in \figurename~\ref{fig:image2}. The subcarrier allocation for this scheme can be expressed as

\begin{equation}
X[k] = 
\begin{cases} 
S[m],  &k = (P+1)m  \\
0, & \text{otherwise}
\end{cases}~,
\end{equation}
where $ m \in \mathbb{Z}^+ $ with $ N_d = \lfloor \frac{N-1}{P+1}\rfloor$, and $ N_{{bp}} = N_d$.

\begin{figure}[t]
    \centering
    \hspace{ 0 cm}
    \includegraphics[width=0.95\linewidth]{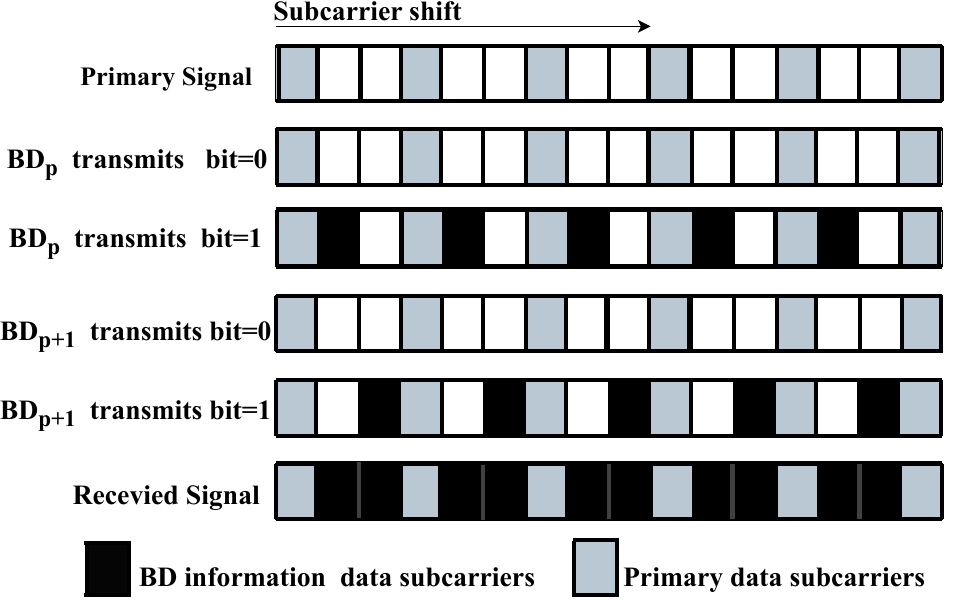}
    \caption{The Fully-Orthogonal scheme with the OFSK modulation.}
    \label{fig:image2}
\end{figure}

\subsubsection{Fully-Orthogonal Scheme with MFSK}
We design the OFDM signal with $2P$ null subcarriers assigned after each data subcarrier throughout the entire \ac{OFDM} symbol for each \ac{BD} to transmit its information. Each \ac{BD} is allocated to a pair of null subcarriers out of the total number of subcarriers between two data subcarriers in a sequence. Each \ac{BD} shifts the incoming OFDM signal to a designated pair of null subcarriers based on its information bit to be transmitted. Specifically, when transmitting bit `0', $\mathrm{BD}_p$ performs a frequency shift corresponding to the first subcarrier of its designated pair of null subcarriers, meanwhile, $\mathrm{BD}_{p+1}$ shifts to the first subcarrier of its designated pair of subcarrier. Conversely, when transmitting bit `1', $\mathrm{BD}_p$ performs a frequency shift corresponding second subcarrier in the designated pair of subcarriers, whereas, $\mathrm{BD}_{p+1}$ shifts to the second subcarrier in designated pair of its null-subcarrier as illustrated in \figurename~\ref{fig:image4}.  
Although this technique reduces interference in the received signal, it compromises spectral efficiency. The allocation of subcarriers for this scheme can be expressed as

\begin{equation}
    X[k] =
\begin{cases} 
S[m]~, & k = (2 P+1)m \\
0~, & \text{otherwise}
\end{cases}~,
\end{equation}
where $ m \in \mathbb{Z}^{+} $, and $ N_d =\lfloor \frac{N}{2P+1}\rfloor$, and $N_{{bp}} =  2N_d$.

\begin{figure}[t]  
    \centering
    \hspace{0 cm}
   \includegraphics[width=0.90\linewidth]{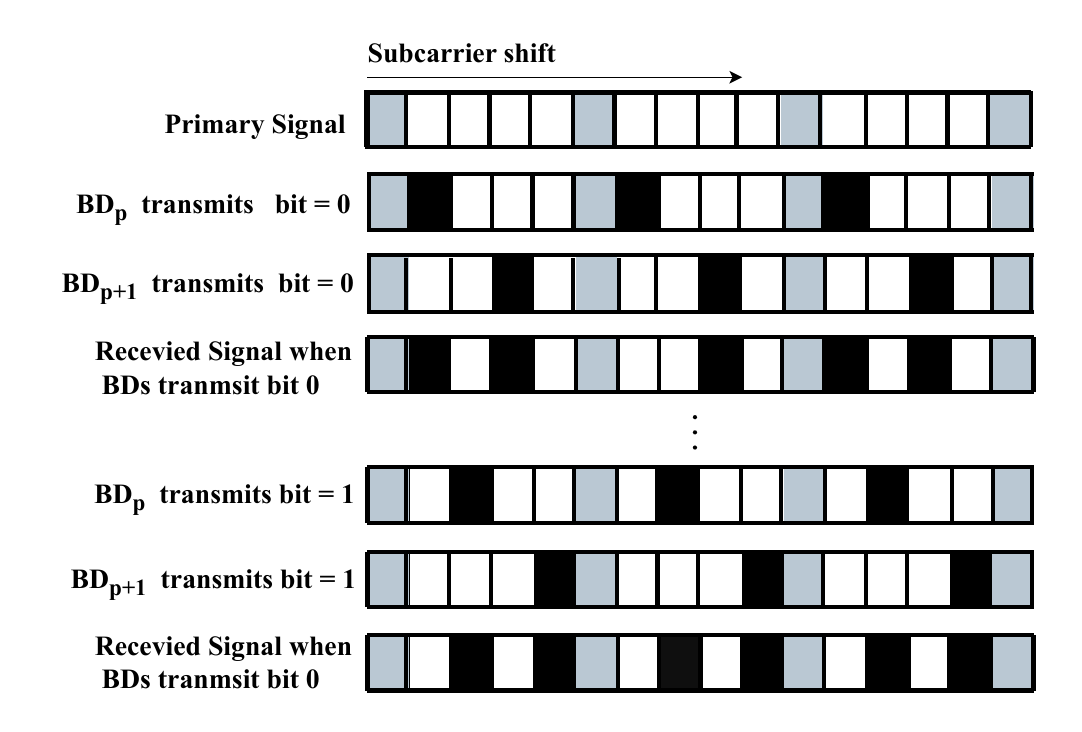} 
    \caption{The Fully-Orthogonal scheme with the MFSK modulation.}
    \label{fig:image4}
\end{figure}

\subsection{Semi-Orthogonal Schemes}

\subsubsection{Semi-Orthogonal Scheme with OFSK}
In this method, we allocate a set of $P$ null subcarriers only after the first primary data subcarrier for each of the \acp{BD} to transmit the \ac{BD} signals in order. Specifically, $\mathrm{BD}_p$ is assigned to the null subcarrier $p$, $\mathrm{BD}_{p+1}$ to the null subcarrier ($p+1$).....$\mathrm{BD}_{P}$ to $P$. For instance, in \figurename~\ref{fig:image3}, $P$ = 2, we pre-allocate $p$ = 1, $(p+1)= P = 2$, null subcarriers after the first primary data subcarrier. This method optimizes spectrum resource utilization by minimizing the number of null subcarriers, offering more support for primary data transmission. However, the subcarrier shifting to non-null designated subcarrier positions results in overlapping between the primary and \ac{BD} signals which introduces partial \ac{DLI}. The subcarrier allocation for this scheme can be expressed as

\begin{equation}
    X[k] =
\begin{cases} 
S[m], & k = m~, \quad m = 0\\
S[m], & k = (P+m), \quad m \geq 1 \\
0, & \text{otherwise}
\end{cases}~,
\end{equation}
where $ m \in \mathbb{Z}^{+}, N_d=N-P$, and $N_{{bp}} = 1$.

\begin{figure}[t]  
    \centering
    \hspace{0 cm}
    \includegraphics[width=0.95\linewidth]{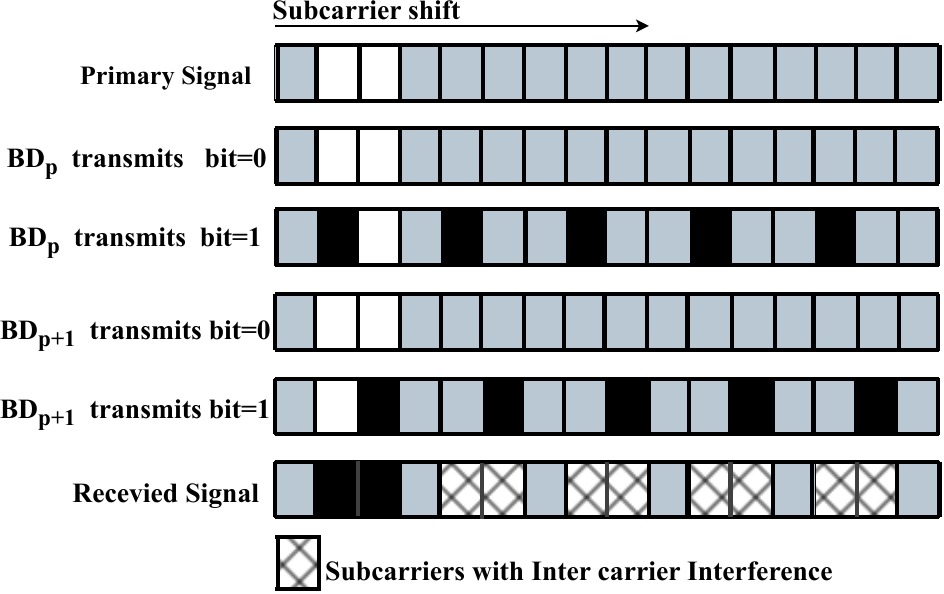} 
    \caption{ The Semi-Orthogonal scheme with OFSK modulation.}
    \label{fig:image3}
\end{figure}

\subsubsection{Semi-Orthogonal Scheme with MFSK}
To enhance the spectral efficiency of \ac{MFSK} modulation, we allocate $2P$ null subcarriers only after the first primary data subcarrier, allowing each \ac{BD} to transmit the \ac{BD} signal within a limited number of designated subcarriers. For instance, for two \acp{BD} case in \figurename~\ref{fig:image5} only four null subcarriers are allocated after the first primary data subcarrier while the remaining subcarriers are reserved for primary data transmission. Although this scheme enhances spectral efficiency, it introduces partial \ac{DLI}, as the shifted \ac{BD} signal may overlap with the primary signal transmission in the non-null designated subcarriers. The allocation of subcarriers for this scheme can be expressed as
\begin{equation}
   X[k] =
    \begin{cases} 
S[m], & k = m~, \quad m = 0\\
S[m], & k = (2P+m), \quad m \geq 1\\
0, & \text{otherwise}
    \end{cases}~,
\end{equation}
where $ m \in \mathbb{Z}^{+}$ with $N_d=N-2P$, and $N_{{bp}}=2$.

\subsubsection{Resolving Partial-Interference}

\section{Detection and Performance Analysis}
In \ac{SR-BC}, the received signal is a superposition of the primary signal directly from the \ac{BS} and the signals transmitted from multiple \acp{BD}. At the \ac{Rx}, \ac{CP} is removed, followed by a discrete Fourier transform (DFT) over the $N$ points. The frequency domain representation of the received signal can then be expressed as
\begin{equation}
\begin{aligned}
        Y[k] = H_{\mathrm{d}}[k]X[k] + \sum_{p=1}^P&\alpha_p (H_{\mathrm{b}p}[k]H_{\mathrm{f}p}[k] X[k]) \\
        &*B_p[k] + W[k]~,
\end{aligned}
\end{equation}
where $Y[k]$ denotes the \ac{OFDM} received symbol on the $k$-th subcarrier while terms $X[k]$, $H_{_\mathrm{d}}[k]$, $H_{\mathrm{b}p}[k]$, $H_{\mathrm{f}p}[k]$ and $W[k]$ represent the frequency domain equivalents of $x[n]$, $h_{_\mathrm{d}}[n]$, $h_{\mathrm{b}p}[n]$, $h_{\mathrm{f}p}[n]$ and $w[n]$, respectively. $B_p[k]$ is the \ac{DFT} of the exponential function $b_p[n]$, which is given by $\delta{(k-p)}$. By performing a convolution operation, the above equation is reduced to
\begin{equation}
\label{Eq:15}
\begin{aligned}
        Y[k] = H_{\mathrm{d}}[k]X[k] + \sum_{p=1}^P&\alpha_pH_{\mathrm{b}p}[k-p]\\
        &H_{\mathrm{f}p}[k-p] X[k-p] + W[k]~.
\end{aligned}
\end{equation}
In the received signal, the primary and \acp{BD} signals can be expressed separately as
\begin{align}
    \hat{Y}_{\mathrm{d}}[\hat{k}_\mathrm{d}] &= H_{\mathrm{d}}[\hat{k}_\mathrm{d}] X[\hat{k}_\mathrm{d}] + W_{\mathrm{d}}[\hat{k}_\mathrm{d}]~, \quad \hat{k}_\mathrm{d} \in \mathcal{K}_\mathrm{d}~, \\
    \hat{Y}_{\mathrm{b}p}[\tilde{k}_p] &= Y[\tilde{k}{_p}] + W_{\mathrm{b}p}[\tilde{k}_p] ~,\quad \tilde{k}_p\in \mathcal{K}_{b_p}~,
\end{align}
where $ \hat{Y}[\tilde{k}{_p}] =\alpha_p H_{\mathrm{f}p}[\tilde{k}_p] X[\tilde{k}_p] H_{\mathrm{b}p}[\tilde{k}_p]$ while $\hat{k}_\mathrm{d}$ and $\tilde{k}_{p}$ represent the subcarriers occupied by the primary signal and \ac{BD} signal from the $\mathrm{BD}_p$, respectively. According to the modulation process in all schemes, $X[\tilde{k}]$ can generally be expressed as
\begin{equation}
X[\tilde{k}] =
\begin{cases} 
0~, & b[n] = 0~, \\ 
\alpha_p X_{\mathrm{d}}[{k}_{p} -(p+q)]~, & b[n] = e^{j 2\pi f_{p+q} n / N}~.
\end{cases}
\end{equation}
where $ q \in \mathbb{Z}^{+}$.

\begin{figure}[t]  
    \centering
    \hspace{ 0 cm}
   \includegraphics[width=0.95\linewidth]{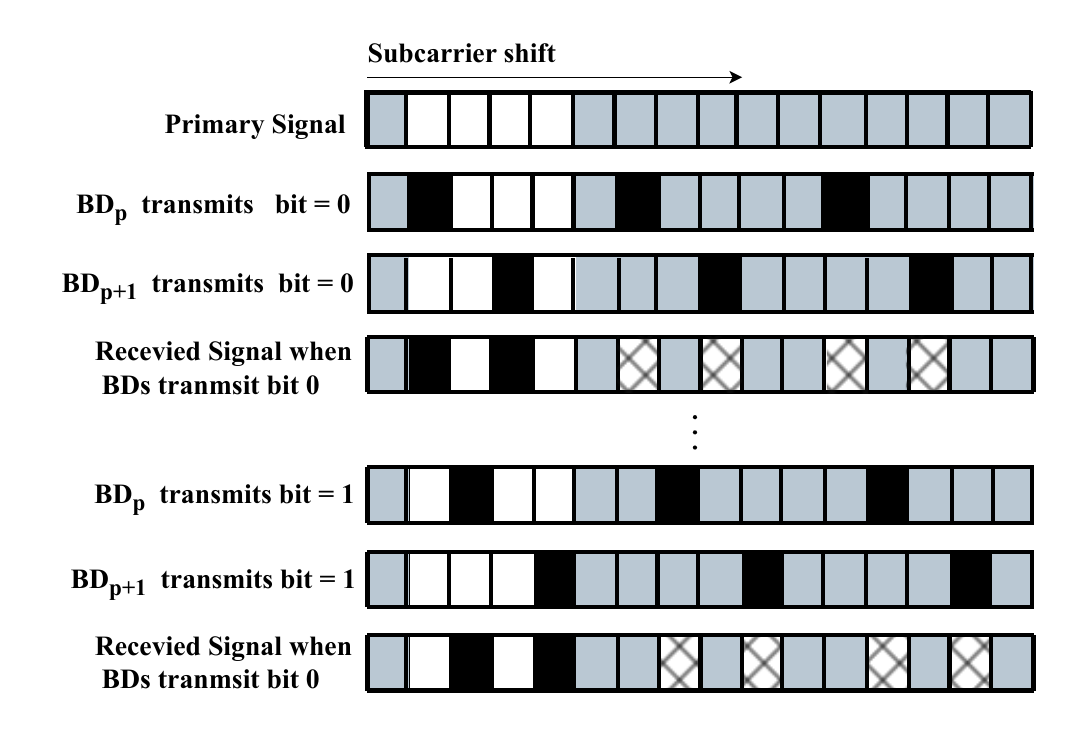} 
    \caption{The Semi-Orthogonal scheme with the MFSK modulation.}
    \label{fig:image5}
\end{figure}

\subsection{Non-coherent Detection for OFSK}
During reception, the \ac{Rx} performs conventional coherent detection to decode its dedicated data, while simultaneously performing non-coherent detection to identify \ac{BD} signal. The non-coherent detection process does not require \ac{Rx} to have phase information or any channel state information \cite{ref35}. This makes the detection of \acp{BD} signals less complex, simpler, and computationally efficient, since it eliminates the need for complex channel estimation. This approach is well-suited for \ac{SR-BC}. The \ac{Rx} detects the \acp{BD} signals in the designated subcarriers by measuring the presence of backscattered energy, which is analyzed by statistical distribution. As the backscattered \ac{OFDM} signal is equivalent to the shifted \ac{OFDM} symbol, therefore, for a large number of $N$, the corresponding signal, $Y_{\mathrm{b}p}[\tilde{k}_p]$, follows a complex Gaussian distribution with circularly symmetric mean and variance, $\sigma^2_{Y_{\mathrm{b}p}} = |\alpha_p|^2 \sum_{\tilde{k}_p=0}^{N_b-1} |H_{\mathrm{f}p}[\tilde{k}_p]|^2 |H_{\mathrm{b}p}|^2$. The \ac{SR-BC} signal is expressed as

\begin{equation}
\hat{Y}_{bp}[\tilde{k}{_p}] =
\begin{cases} 
W_{\mathrm{b}_p}{[\tilde{k}_p]}, & b_p =  0 \\
{Y}_{bp}[\tilde{k}{_p}] + W_{\mathrm{b}p}{[\tilde{k}_p]}, & b_p =  1
\end{cases}~,
\end{equation}
where $W_{\mathrm{b}p}{[\tilde{k}_p]} \sim \mathcal{C}\mathcal{N}(0, \sigma^2_{W_{\mathrm{b}p}})$. In this context, two hypotheses for the likelihood of detection of $\mathrm{BD}_p$, are tested using a test statistic which is based on the squared sum of the received signal across the designated subcarriers. This test statistic is expressed as

\begin{equation}
\label{Eq:20}
{R_p} = \sum_{\tilde{k}_p=0}^{N_\mathrm{b}-1}| \hat{Y}_{bp}[\tilde{k}{_p}]| ^2~,
\end{equation}

\begin{itemize}
    \item \text{Probability distribution under $H_0$, $Pr({R_p}|B_p=0)$}
    
The test statistic $R_p$ follows a Gaussian noise distribution, $W_{bp}[\tilde{k}_p] \sim \mathcal{CN}(0, \sigma^2_{W_{bp}})$, indicating the absence of backscatter signal from $\mathrm{BD}_p$ which corresponds to the transmission of an information bit ‘0’, where the receiver detects only noise. The test statistic $R_p = R_p^0$, representing the noise energy, modeled as a complex circularly symmetric Gaussian variable with independent, identically distributed (i.i.d.) real and imaginary parts, each with a mean of $0$ and variance $\sigma^2_{W_{bp}}$. Consequently, $|W_{\mathrm{b}p}[\tilde{k}_p]|^2$ follows an exponential distribution. The energy $R_p^0$ is the sum of $N_\mathrm{b}$ independent exponential random variables. Its probability density function (PDF) is challenging to compute directly, as it requires performing $N_\mathrm{b} -1$ convolution operations, equivalent to computing $N_\mathrm{b} -1$ integrals. For a given ${R_p} = {R}_0$, the characteristic function of $|{W}_{bp}{[\tilde{k}_p]}|^2$ is given by
\begin{equation}
\Phi_{R_p^0}(t) =\mathbb{E} \left[ e^{itx} \right] =\prod_{\tilde{k}_{p}=1}^{N_\mathrm{b}} \frac{1}{1 - it\lambda_{\tilde{k}_{p}}^{-1}}~,
\end{equation}
where $\lambda_{\tilde{k}_{p}}^{-1} = 2\sigma_{w_{\mathrm{b}p},{\tilde{k}{_p}}}^2$.
Then the \ac{PDF} of ${R_p^0}$, i.e., $Pr(R_p|B_p=0)$ is obtained by taking the inverse Fourier transform.

\begin{equation}
f_{R_p^0}(x) =  \int_{-\infty}^\infty \frac{1}{2\pi jt}\Phi_{R_p^0}(t) e^{-jtz} \, dt~.
\end{equation}

\item \text{Probability distribution  under $H_1$, $Pr({R_p}|B_p=1)$}

In this case, the test statistic ${R_p}$ = ${R_p^1}$, is influenced by both noise and the backscattered energy from the $\mathrm{BD}_p$,  which corresponds to the transmitted information bit `1'. The received signal depends on the two cascaded channels forward channel and backscatter channel, denoted by $h_{\mathrm{f}p}[n]$  and $h_{\mathrm{b}p}[n]$, respectively.
\end{itemize}

\begin{lemma}
\label{lemma:1}
 Due to the short distance between $\mathrm{BD}_p$ and the \ac{Rx} hold assumption that, the signal through $h_{\mathrm{b}p}[n]$ is transmitted through a single path and remains constant during one \ac{OFDM} symbol, meaning that the cascaded channels $h_{\mathrm{b}p}[n]*h_{\mathrm{f}p}[n]$ can be approximated constant. For each backscatter link, the squared magnitude $|h_{\mathrm{b}p}|^2=v_{p}^2$ each follows a Rayleigh distribution. The energy $ R_p = {R_p^1} $ includes the noise and the backscatter-modulated signal, where the characteristic function of $|{W}_{bd}{[\tilde{k}_p]}|^2$ under the given $ v_{p}$.
\begin{equation}
\Phi_{{R_p^1}|v_p}(t) = \prod_{\tilde{k}{_p}=1}^{N_\mathrm{b}} \frac{1}{1 - it\lambda_{\tilde{k}{_p}}^{-1}}~,
\end{equation}
\textit{where}
\begin{equation}
\lambda_{\tilde{k}{_p}}^{-1} = 2 \alpha_\mathrm{b} v_{p}^2 \sum_{\tilde{k}{_p}=0}^{N_\mathrm{b}-1} \left(\sigma_{h_{p},\tilde{k}_p}^2 + \sigma_{w_{\mathrm{b}p}}^2 \right)~,
\end{equation}
\textit{$ \sigma_{h_{p},{\tilde{k}{_p}}}^2 $ is the power of the cascaded channel. The PDF is then obtained as}
\begin{equation}
f_{{R_p^1}|v_{p}}(x) = \int_{-\infty}^{\infty} \frac{1}{2 \pi j t} \Phi_{{R_p^1}|v_{p}}(t) e^{-j t z} dt~.
\end{equation}
\end{lemma}
\begin{proof}
The proof of Lemma~\ref{lemma:1} is given in Appendix~\ref{prof:lemma:pdfOFSK}.
\end{proof}

\paragraph{Probability of Error}
We measure and evaluate the performance of the non-coherent detector in terms of the probability of error ($P_{e}$), defined as the sum of the average \ac{PFA} and average \ac{PMD}. For the $P$ \acp{BD}, this probability of error can be expressed as  
\begin{align}
P_{e} &= \frac{1}{P} \sum_{p=1}^{P}\left(  \frac{1}{2}\text{PMD}_p 
+ \frac{1}{2}\text{PFA}_p \right)~. 
\end{align}
Considering equal prior probabilities for the hypotheses $H_0$ and $H_1$, the probability of false alarm ($\text{PFA}_p$), which represents the likelihood of detecting a signal when $B_p = 0$ for a given threshold value $\gamma$, is expressed as  
\begin{align}
\text{PFA}_p(\gamma) &= Pr({R_{p}} > \gamma \mid B_p= 0)~, \notag\\
             &= 1 - Pr({R_{p}} \leq \gamma \mid B_p= 0) = 1 - F_{R_{p}}(\gamma)~,
\end{align}
where $F_{R_{p}}(\gamma)$ is the cumulative distribution function (CDF) of ${R_p}$. Similarly, the probability of missed detection ($\text{PMD}_p$), which represents the likelihood of failing to detect a signal when $B_p = 1$, is given by  
\begin{equation}
\text{PMD}_p(\gamma) = Pr({R_{p}} \leq \gamma | B_p= 1) = F_{R_{p}}(\gamma)~,
\end{equation}
where $F_{R_{p}}(\gamma)$ is the same CDF of $R_p$, but evaluated under the alternative hypothesis.

\paragraph{Threshold Optimization}  
The decision threshold ($\gamma$) determines the balance between $\text{PFA}_p$ and $\text{PMD}_p$. It is chosen to minimize the probability of error $P_{e}$. The optimal threshold $\tilde{\gamma}$ is given by  
\begin{equation}
\tilde{\gamma} = \arg \min_\gamma P_{e}(\gamma)~.
\end{equation}

For a fixed $\text{PFA}_p$, minimizing $\text{PMD}_p$ requires searching for $\tilde{\gamma}$ using numerical optimization techniques, such as a one-dimensional linear search. This ensures that $P_{e}$ is minimized for a given channel variance $\sigma^2_{W_{\mathrm{b}p}}$. For a given $\mathrm{BD}_p$, if $v_p$, representing the magnitude of the backscatter channel coefficient $h_{\mathrm{b}p}[n]$ following the Rayleigh distribution, the CDF $F_{R_{p}}(\gamma)$ can then be expressed as

\begin{equation}
F_{R_{p}}(\gamma) = \int_{0}^{\infty} F_{R_{p}}(\gamma; v_p) f(v_p) \, dv_p~,
\end{equation}
where $F_{R_{p}}(\gamma; v_p)$ is the conditional CDF of ${R_{p}}$ given $v_p$, $f(v_p) = \frac{v_p}{\sigma_{v_p}^2} e^{-v_p^2 / \sigma_{v_p}^2}$ is the probability density function (PDF) of the Rayleigh distribution. Using the inversion formula from Fourier analysis in \cite{ref36}, $F_{R_{p}}(\gamma; v_p)$ is computed as
\begin{equation}
F_{R_{p}}(\gamma; v_p) = \frac{1}{2} - \int_{-\infty}^\infty \frac{1}{2\pi j t} \Phi_{{R_{p}}|v_p}(t) e^{-j t z} \, dt~.
\end{equation}
In general, $F_{R_{p}}(\gamma)$ lacks a closed-form expression, requiring numerical integration techniques for accurate computation. Alternatively, it can be evaluated using software that supports integration over a Rayleigh-distributed $v_p$. When $v_p$ varies due to multipath effects, $F_{R_{p}}(\gamma)$ is marginalized over $v_p$. However, if $\mathrm{BD}_p$ is close to the receiver, meaning a single-tap path, $v_p$ can be treated as a constant. This simplification makes the computation of $\text{PMD}_p$ and $\text{PFA}_p$ more straightforward.  

\subsection{Non-coherent Detection for MFSK}

In the case of \ac{MFSK} the \ac{Rx} computes the test statistics metrics ${R_p^0}$ for the information bit `0' when assigned $\mathcal{K}_{b_p}^{0}$ set of designated subcarriers and ${R_p^1}$ for the information bit `1', when assigned $\mathcal{K}_{bp}^{1}$ sets of designatedsubcarriers.

\begin{align}
\label{Eq:32}
    {R_p^0} &= \sum_{\tilde{k}_p \in \mathcal{K}_{bp}^{0}} |Y[\tilde{k}_p]|^2, &
    {R_p^1} &= \sum_{\tilde{k}_p \in \mathcal{K}_{bp}^{1}} |Y[\tilde{k}_p]|^2.
\end{align}

The \ac{Rx} then determines the $\mathrm{BD}_p$'s transmitted bit by comparing ${R_p^0}$ and ${R_p^1}$.

\begin{equation}
    b_p =  \begin{cases}
        0, & {R_p^0} > {R_p^1}\\
        1, & {R_p^0}< {R_p^1}
    \end{cases}~,
\end{equation}
where $b_p$ represents the detected bit such that, if ${R_p^0}$ exceeds ${R_p^1}$, the detected bit is `0', the BD shifts the primary signal to the subcarrier set $\mathcal{K}_{bp}^{0}$. Conversely, if ${R_p^0}$ becomes less than ${R_p^1}$ the \ac{Rx} detects information bit `1', the $\mathrm{BD}_p$ shifts the primary signal to $\mathcal{K}_{bp}^{1}$. However, detection errors can occur due to the influence of channel conditions and noise such that, the $\mathrm{BD}_p$ transmits information bit `0' but ${R_p^0}$ less than ${R_p^1}$, the bit erroneously detected as `1', and vice versa. 

The probability of error for non-coherent detection in \ac{MFSK} modulation technique is expressed as
\begin{align}
\label{Eq:34}
P_{e} &= \frac{1}{P} \sum_{p=1}^{P} \frac{1}{2}( ( \Pr(b_p= 1 \mid B_p=0) + \Pr(b_p= 0 \mid B_p=1))~
\end{align}
where $\Pr(b_p= 1 | B_p=0)$ represents the average \ac{PMD}, and $\Pr(b_p= 1 | B_p=1)$ denotes the average \ac{PFA}. These probabilities can be expressed as
\setcounter{equation}{34} 
\begin{equation} 
    \Pr({b_p}|B_p) = 
    \begin{cases}
        \Pr({R_p^0} - {R_p^1} \leq 0), & b_p= 0 \\
        \Pr({R_p^0} - {R_p^1} > 0), & b_p= 1
    \end{cases} \tag{\theequation a}~, 
\end{equation}

\begin{equation} 
    \Pr({b_p}|B_p) = 
    \begin{cases}
        F_{{R_p^0} - {R_p^1}}(0), & b_p= 0 \\
        1 - F_{{R_p^0} - {R_p^1}}(0), & b_p= 1
    \end{cases} \tag{\theequation b}~, 
\end{equation}  
The distributions of ${R_p^0}$ and ${R_p^1}$ depend on the cascaded channels $h_{\mathrm{f}p}[n]$  and $h_{\mathrm{b}p}[n]$. For a specific $h_{\mathrm{b}p}[n]$, the probability density function (PDF) of the cascaded backscatter channel can be derived from Lemma~\ref{lemma:1}, Further, the \ac{CDF} of ${R_p^0} - {R_p^1}$ can also be represented by $F_{{R_p^0} - {R_p^1}}(0)$, as shown below. 
\begin{theorem}
\label{theorem:1}
 The \ac{CDF} of $ {R_p^0} - {R_p^1} $ can be computed as
\begin{equation}
F_{{R_p^0} - {R_p^1}}(0) = \int_0^{\infty} F_{{R_p^0} - {R_p^1}}(0; v_p) f(v_p) \, dv_p~,
\end{equation}
where
    \begin{equation}
        F_{{R_p^0} - {R_p^1}}(0; v_p) = \frac{1}{2} -  \int_{-\infty}^{\infty} \frac{\Phi_{{R_p^0}|v_p}(t) \Phi_{{R_p^1}|v_p}^*(t)}{{2\pi j}t} \, dt~,
    \end{equation}
with $v_p = |h_{\mathrm{b}p}[n]|$ representing the backscatter channel gain and $f(v_p) = \frac{v_p}{\sigma_{v_p}^2} \exp\left(-\frac{{v_p}^2}{\sigma_{v_p}^2}\right)$. The characteristic functions $\Phi_{{R_p^0}|v_p}(t)$ and $\Phi_{{R_p^1}|v_p}(t)$ which are given by
\begin{align}
    \Phi_{{R_p^0}|v_p}(t) &= \frac{1}{1 - it\lambda_{\tilde{k}{_p}}^{-1}}~, &
    \Phi_{{R_p^1}|v_p}(t) &= \frac{1}{1 - it\lambda_{\tilde{k}{_p}}^{-1}}~,
\end{align}
respectively,
where $ \lambda^{-1}_{\tilde{k}} = {2\alpha_\mathrm{b} v_p^2} \sum_{\tilde{k}_p=0}^{N_\mathrm{b}-1} \sigma^2_{h_p, \tilde{k}_p} + \sigma^2_{W_{\mathrm{b}p}} $.
\end{theorem}
\begin{proof}
    The proof of Theorem~\ref{theorem:1} is provided in Appendix~\ref{proof:theorem:pdfmfsk}.
\end{proof}
The derivations are applicable to cascaded backscatter channels, assuming Rayleigh fading for both forward and backward links. To calculate $F_{{R}{b_p^{0}} - {R}{b_p^1}}(0)$, we marginalize over $v_p$ to compute $\Pr(b | B_p)$ in (\ref{Eq:34}). Although a closed-form expression is not available, the result remains analytically tractable and can be evaluated numerically to facilitate the calculation of $ P_e$). 

The aforementioned non-coherent detection methods apply to both Fully-Orthogonal and Semi-Orthogonal schemes for detecting \ac{BD} information using null subcarriers. However, in the Semi-Orthogonal scheme, the allocation of null subcarriers differs from that in the Fully-Orthogonal scheme, resulting in a limited number of null subcarriers available for \ac{BD} detection. As (\ref{Eq:20}) and (\ref{Eq:32}) indicate the non-coherent detector accounts for the number of null subcarriers; thus, a reduction in their quantity adversely impacts the detector's performance. The modulation techniques employed by \acp{BD} shift the entire incoming \ac{OFDM} symbol by a unique subcarrier shift, and during reception, these pre-allocated null subcarriers exclusively contain the \ac{BD} signal. However, the remaining designated subcarriers in the \ac{OFDM} signal contain both \ac{BD} data and primary signal data superimposed. To address this issue, we employ \ac{SIC} to remove the unwanted signal components. The receiver first estimates $H_{\mathrm{d}}[k]$ and $X[k]$ in (\ref{Eq:15}), then removes the direct transmission component, obtaining the residual signal $Y[k]_{res}$. Due to estimation errors $\epsilon_{h}= H_{\mathrm{d}}[k] - \hat {H_{\mathrm{d}}}[k]$ and $\epsilon_{x}= X[k] - \hat {X[k]}$, the residual signal contains backscattered components along with residual interference and noise.
\begin{equation}
\label{Eq:38}
\begin{aligned}
Y[k]_{res} =  \hat{Y}_{\mathrm{b}p}[\tilde{k}_p] + \epsilon_{h} \hat{X}[k] + \epsilon_{x} \hat{H_{\mathrm{d}}}[k] + W[k]~,
\end{aligned}
\end{equation}
where $W[k]\sim \mathcal{CN}(0, \sigma^2)$. The \acp{BD} symbols $B_p[k]$ are then detected using maximum likelihood (ML) estimation:
\begin{equation}
    \hat{B}_{b_p}[k] = \arg \min_{\tilde{B}_{p}[k]} \left| Y[k]_{\text{res}} - \alpha_p H_{\mathrm{b}p}[k] H_{\mathrm{f}p}[k] X[k] {\tilde{B}_p[k]} \right|^2~.
\end{equation}

\subsection{ Sum-Rate of \ac{SR} System}

Sum-rate is a critical metric for evaluating spectral efficiency in \ac{OFDM}-based \ac{SR} systems, particularly in the context of assessing various subcarrier allocation strategies. Since the \ac{SR} operates by sharing the spectrum between the primary and \acp{BD} signals, the total sum-rate is determined by the achievable data rates of both. Several factors, including reflection coefficients, \ac{IBDI}, the number of \acp{BD}, and \ac{DLI}, significantly influence the achievable sum-rate. As each \ac{BD} experiences independent fading, the sum-rate reflects the overall system performance rather than that of individual \ac{BD}. It accounts for transmission opportunities for both primary and \ac{BD} signal transmissions, offering insights into interference effects, spectral resource allocation, and \acp{BD} deployment \cite{ref37}. To analyze the total sum-rate under different subcarrier allocation strategies and multiple \acp{BD}, we decompose it into the sum-rate of the primary and that of \ac{BD} signal transmissions, with the total sum-rate being their sum.
\subsubsection{Sum-rate of \ac{SR} with Fully-Orthogonal Schemes}
In this scheme, the orthogonality of \ac{BD}-to-\ac{BD} subcarrier allocation and the separation between \acp{BD} and the primary signal transmission effectively eliminate interference resulting in an interference-free \ac{SR-BC} with enhanced \ac{BER} performance, as discussed in \cite{ref25}. However, restriction in subcarrier reuse among the transmissions causes poor overall spectral efficiency, leading to a reduction in the overall achievable sum-rate. For this scheme the $\mathrm{BD}_p$ transmission sum-rate can be expressed as
\begin{equation}
R_\text{bd}^\text{FO} = \sum_{p=1}^{P} \sum_{\tilde{k}_p \in \mathcal{K}_{b_p}^{\mathrm{F_{IF}}}} 
{\Delta f} \log_2 \left( 1 + \frac{{\alpha_p^2 |H_{\mathrm{b}p}[ \tilde{k}_p]|^2 |H_{\mathrm{f}p}[ \tilde{k}_p]|^2}} {\sigma^2_{W}}\right),
\end{equation}
where $\Delta f$ is the subcarrier spacing, and  $\mathcal{K}_{b_p}^{\mathrm{F_{IF}}}$ is a set of null subcarriers which is allocated based on applied modulation techniques, \ac{OFSK} or \ac{MFSK}. To obtain the total sum-rate of the system, we need to also calculate the sum-rate of the primary signal transmission. In the Fully-Orthogonal scheme, the primary signal transmission does not encounter interference from the \ac{BD} signals as shown in \figurename~ \ref{fig:image2} and \figurename~ \ref{fig:image4}. Thus, we can describe the primary sum-rate as

\begin{equation}
    R_\text{Rx}^\text{FO} =  
\sum_{\hat{k}_p \in \mathcal{K}_{d}^{\mathrm{F_{IF}}}} \Delta f \log_2 \left( 1 + {\psi_{_\mathrm{F}}} \right)~,
\end{equation}
where $\psi_{_\mathrm{F}}$ =$\frac{|H_{\mathrm{d}}[\hat{k}_p]|^2}{ \sigma^2_{W_{d}}}$ is the \ac{SNR} and $\mathcal{K}_{d}^{\mathrm{F_{IF}}}$ is set of primary data subcarriers. Now, we can calculate the total sum-rate of the \ac{OFDM}-based fully-orthogonal \ac{SR-BC} scheme as 
\begin{equation}
    R_\text{total}^{FO} = R_\text{bd}^\text{FO}+R_\text{Rx}^\text{FO}~.
\end{equation}

\subsubsection{Sum-rate of \ac{SR} with Semi-Orthogonal Schemes}

In the Semi-Orthogonal scheme, limited null subcarriers are allocated to ensure orthogonality, while the remaining designated subcarrier positions experience interference due to subcarrier shifts induced by modulation technique applied at $\mathrm{BD}_p$, as illustrated in \figurename~\ref{fig:image3} and \figurename~\ref{fig:image5}. This shift leads to subcarrier reuse, which destroys orthogonality but enhances transmission opportunities for both primary and \ac{BD} signal transmissions, thereby increasing the total achievable sum rate. In this scheme, the sum rates of \acp{BD} and primary signal transmission are determined based on subcarrier occupancy while accounting for interference effects. For the designated subcarrier left empty, the sum-rate of \acp{BD} transmission is expressed as
\begin{equation}
R_\text{bd}^{S_\text{IF}} = \sum_{p=1}^{P} \sum_{\tilde{k}_p \in \mathcal{K}_{b_p}^{S_\text{IF}}} 
{\Delta f} \log_2 \left( 1 + \frac{{\alpha_p^2 |H_{\mathrm{b}p}[ \tilde{k}_p]|^2 |H_{\mathrm{f}p}[ \tilde{k}_p]|^2}} {\sigma^2_{W}}\right),
\end{equation}
where $\mathcal{K}_{b_p}^{S_\text{IF}}$ is set of null subcarriers.
For the designated subcarriers shared with the primary  signal transmission, the sum-rate is expressed as
\begin{equation}
R_\text{bd}^{S_\text{I}} = \sum_{p=1}^{P} \sum_{\tilde{k}_p \in \mathcal{K}_{b_p}^{S_\text{I}}} 
{\Delta f} \log_2 \left( 1 + \frac{{\alpha_p^2 |H_{\mathrm{b}p}[ \tilde{k}_p]|^2 |H_{\mathrm{f}p}[ \tilde{k}_p]|^2}} {I_{p}[ \hat{k}_p]+\sigma^2_{W}}\right),
\end{equation}
where $\mathcal{K}_{b_p}^{S_\text{I}}$ is a set of designated subcarriers shared between primary and $\mathrm{BD}_p$ transmissions and $I_{p}[ \hat{k}_p] = {|H_{\mathrm{d}}[\hat{k}_p]|^2}$ represents the interference causing a loss of orthogonality in the Semi-Orthogonal scheme. However, by applying \ac{SIC}, this interference is suppressed, leaving only residual interference, which can be derived from (\ref{Eq:38}). The cumulative sum-rate of $BD_p$ transmission is expressed, then sum-rate can be expressed as
\begin{equation}
    R_\text{bd}^{SO} = R_\text{bd}^{S_\text{IF}}+R_\text{bd}^{S_\text{I}}~.
\end{equation}
In all cases $\mathcal{K}_{b_p}^{S_\text{IF}}$ and  $\mathcal{K}_{b_p}^{S_\text{I}}$  are allocated based on the applied modulation techniques, \ac{OFSK} or \ac{MFSK}. To obtain the total sum-rate of the Semi-Orthogonal scheme, we have also to calculate the sum-rate for the primary signal transmission, which is expressed as
\begin{equation}
R_\text{Rx}^\text{SO} =  
\sum_{\hat{k}_p \in \mathcal{K}_{d}^{S_\text{IF}}} \Delta f \log_2 \left( 1 + \psi_{_S} \right) +  
\sum_{\hat{k}_p \in \mathcal{K}_{d}^{S_\text{I}}} \Delta f \log_2 \left( 1 + \lambda \right)~,
\end{equation}
where  $\lambda =\frac{|H_{\mathrm{d}}[\hat{k}_p]|^2}{\alpha_p^2 |H_{\mathrm{b}p}[\tilde{k}_p]|^2 |H_{\mathrm{f}p}[\tilde{k}_p]|^2+ \sigma^2_{W_{d}}}$ is the interference from $\mathrm{BD}_p$ transmissions and $\psi_{_S}$ =$\frac{|H_{\mathrm{d}}[\hat{k}_p]|^2}{ \sigma^2_{W_{d}}}$ is the \ac{SNR} corresponding to primary signal transmission in the Semi-Orthogonal scheme. The total sum-rate of the \ac{OFDM}-based \ac{SR-BC} in the Semi-Orthogonal scheme is expressed as 
\begin{equation}
    R_\text{total}^{SO} = R_\text{bd}^\text{SO}+R_\text{Rx}^\text{SO}~.
\end{equation}

\subsection{SBC Integration with Existing Technologies}
In the proposed scheme, the primary signal and the \acp{BD} signal designs are based on established and readily deployed wireless standards such as Wi-Fi, LTE, and 5G, facilitating direct compatibility with the existing communication infrastructure \cite{ref33}. Moreover, for effective modulation of the primary signal, both the \ac{BD} and the \ac{Rx} must have prior knowledge of \ac{BS} transmission and must maintain synchronization. The \ac{Rx}, acting as a standard user, obtains the necessary details about the primary signal and aligns with the BS using control signals such as the \ac{PSS} and \ac{SSS}, as described in the LTE and 5G specifications. Conversely, \acp{BD} might need supplementary control signals from the \ac{BS}, accompanied by specific synchronization techniques, such as those leveraging the periodic characteristics of LTE signals, as described in \cite{ref37}. Although the proposed methods facilitate the integration of the \ac{SR-BC} to most of the existing technologies in wireless communication, as a frequency shift approach in \ac{OFDM}, it could be susceptible to \ac{CFO}, which creates a large impact on the quality of the received signal, as explained in the next section.

\subsection{CFO Estimation and Compensation in Multi-\ac{BD} \ac{SR-BC}}
\ac{OFDM} is highly susceptible to synchronization errors, notably \ac{CFO} errors typically stem from oscillator mismatches between the transmitter and receiver, which can lead to a fractional or integer shift of subcarriers, resulting in miss detection. Therefore, for effective application of the \ac{OFSK} and \ac{MFSK} schemes, precise estimation and compensation of \ac{CFO} are necessary in order to maintain optimal system performance.
Various techniques for \ac{CFO} estimation have been extensively discussed in the literature \cite{ref38,ref39} where the methods for \ac{CFO} estimation in \ac{OFDM} systems are examined in which the pilot-based mean-square error estimation technique demonstrates superior performance compared to other approaches. 
We consider the transmission of $M$ symbols of the designed \ac{OFDM} signal from the \ac{BS}, with \ac{CP} length denoted as  $N_{cp}$  and the pilot subcarriers indexes $u_{\hat{k}_p}$. 
In the presence of \ac{CFO}, after discarding the \ac{CP}, the received \ac{OFDM} symbol corresponding to the $i^{th}$, transmitted symbol is given by
\begin{align}
    \hat{Y}_{\mathrm{d}}
[\hat{k}_\mathrm{d}] &= e^{j 2 \pi \varepsilon_0 \left( N_{cp} + N \right) / N} 
    \mathbf{\theta_p}(\varepsilon_0) 
    F_N^{H_{\mathrm{d}}} 
    H_{\mathrm{d}} 
    X_i[\hat{k}_\mathrm{d}] 
    + W_{\mathrm{d}}[\hat{k}_\mathrm{d}]~,
\end{align}
    and
\begin{align}
     \hat{Y}_{\mathrm{b}p}[\tilde{k}_p] &= \sum_{p=1}^P 
    e^{j 2 \pi \varepsilon_0 \left( N_{cp} + N \right) / N} 
    \mathbf{\theta_p}(\varepsilon_0) 
    F_N^{H_{\mathrm{f}p}} 
    H_{\mathrm{f}p}[\tilde{k}_p]{X}_i[\hat{k}_p]\notag\\
     &\quad H_{\mathrm{b}p}[\tilde{k}_p]  
    *B_p[k]\alpha_p  
    + W_{\mathrm{b}p}[\tilde{k}_p]~,
\end{align}
where
$i \in \{1, 2, \dots, M\}$ and
\begin{equation}
\mathbf{\theta_p}(\varepsilon_0) = \operatorname{diag} \left( 1, e^{j 2 \pi \varepsilon_0 / N}, \dots, e^{j 2 \pi \varepsilon_0 (N-1) / N} \right)~.
\end{equation}  
At the receiver, the \ac{CFO} is estimated using following optimization problem,
\begin{equation}
\bar{\varepsilon} = \arg\max_{\varepsilon_0} \left\{ Y_i^H \mathbf{\theta_p}(\varepsilon_0) X_i \left( X_i^H X_i \right)^{-1} X_i^H \mathbf{\theta_p}^H(\varepsilon_0)Y_i \right\}~.
\end{equation} 
After estimating $\varepsilon_0$, \ac{CFO} compensation is applied, the received symbol becomes
\begin{equation}
\bar{\mathbf{Y[k]}}_i = e^{j 2 \pi \bar{\varepsilon} \left( N_{cp} + N \right) / N} \mathbf{\theta_p}^H(\bar{\varepsilon}) \mathbf{Y[k]}_i~.
\end{equation}  
The estimation process  may account for some errors leading to residual \ac{CFO}, denoted as
\begin{equation}
\Delta \varepsilon = \varepsilon_0 - \bar{\varepsilon}~, 
\end{equation}  
which would cause phase shifts across subcarriers, leading to \ac{ICI}. The optimal \ac{CFO} estimate minimizes the residual \ac{CFO}, which is analyzed using variance calculations of pilot subcarriers \cite{ref39}.
\begin{equation}
F(\bar{\varepsilon}) = \frac{1}{N_p} \sum_{k=1}^{N_p} 
\frac{\sum\limits_{i=0}^{M-1} \left| \frac{\bar{Y}[k]_i (u_{\hat{k}_p}; \bar{\varepsilon})}{P_i(k)} \right|^2 
- \left| \frac{1}{M} \sum\limits_{i=0}^{M-1} \frac{\bar{Y}[k]_i (u_{\hat{k}_p}; \bar{\varepsilon})}{P_i(k)} \right|^2}
{\sum\limits_{i=0}^{M-1} \left| \frac{\bar{Y}[k]_i (u_{\hat{k}_p}; \bar{\varepsilon})}{P_i(k)} \right|^2}~.
\end{equation}
The optimal estimated \ac{CFO} is obtained by minimizing this function  
\begin{equation}
\epsilon = \arg \min_{\bar{\epsilon}} F(\bar\epsilon)~.
\end{equation}

\section{Numerical Results}
In this section, we present a comprehensive simulation-based evaluation of the proposed schemes, focusing on key performance metrics such as error-rate and sum-rate of the \ac{SR} system.  Initially, we examine the impact of $\alpha_p$, under different \ac{SNR} for the transmission of the \acp{BD} signals. This parameter plays a crucial role in determining the quality of the backscattered signal. Next, we explore the impact of \ac{OFDM} signal characteristics. Specifically, we investigate the impact of the size of $N$, and subcarrier allocation schemes, which altogether are fundamental factors in determining the reliability and spectral efficiency of the \ac{SR} system. After conducting a detailed analysis of the proposed work, we benchmark our results against conventional techniques used in multi-IoT device scenarios, as outlined in \cite{ref26}, while ensuring compliance with LTE standards as considered in the LTE-based system described in \cite{ref40}. 

In this work, we consider a primary system that employs binary phase shift keying (BPSK) modulation, with $N\in\{64, 128, 256, 512\}$, \ac{CP} length, $N_{cp} = N/8$ and subcarrier spacing, $\Delta f$ =15 kHz. Furthermore, we evaluate system performance for different values of $\alpha_p$, where $\alpha_p \in\{0.25, 0.5, 0.75, 1\}$. The detection threshold $\tilde{\gamma}$ is selected based on the guidelines in \cite{ref9}, with $\text{PFA}_p$ of $10^{-3}$. 

\subsection{Simulation Results for Error Rate}

In Fig.~\ref{fig:image20}, the BER performance of the primary signal transmission is analyzed under four different modulation techniques. The results indicate a significant reduction in \ac{BER} as \ac{SNR} increases, demonstrating that, the Fully-Orthogonal scheme with \ac{MFSK} exhibits the most favorable performance, due to the high degree of Orthogonality between \ac{BD} and Primary signals because of a large number of allocated null subcarriers. In contrast, the Semi-Orthogonal with \ac{OFSK} consistently yields the highest \ac{BER}, indicating greater susceptibility to interference and reduced detection efficiency due to the fewer number of null subcarriers allocated, which signifies a high loss of orthogonality. In particular, at an SNR of 20 dB, the Fully-Orthogonal scheme with \ac{MFSK} achieves a \ac{BER} below $10^{-2}$, whereas the Semi-Orthogonal with \ac{OFSK} remains above $10^{-1}$.

\begin{figure}[t]  
    \centering
   \includegraphics[width=0.85\linewidth]{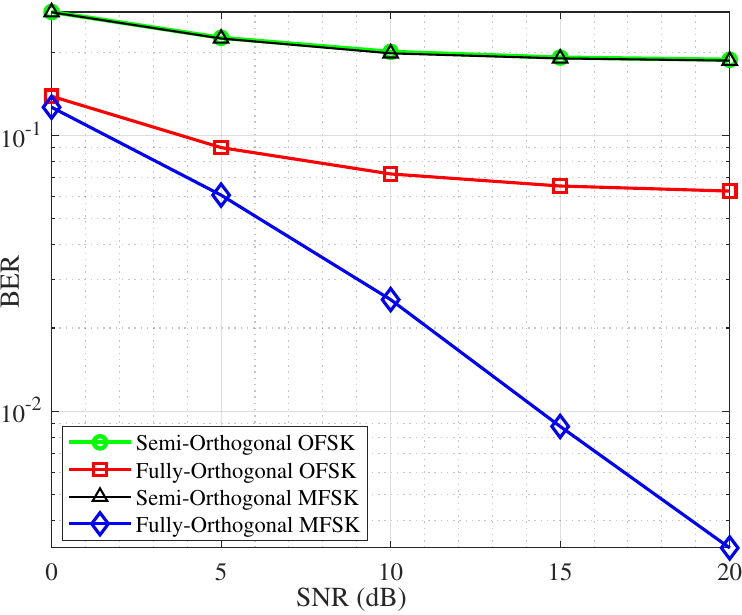} 
    \caption{The BER performance of the primary signal transmissions at $\alpha_p =0.25$.}
    \label{fig:image20}
\end{figure}


In \figurename~\ref{fig:image6}, we analyze the performance of non-coherent detection of \ac{OFSK} modulated signal in Fully-Orthogonal scheme, we examine average \ac{PMD} versus \ac{SNR} for different values of $\alpha_p$. A decrease in the average \ac{PMD} is observed with the increase in \ac{SNR} values. The lowest value of the average \ac{PMD} is achieved at $\alpha_p = 1$, indicating that $\mathrm{BD}p$ reflects the incoming signal with maximum power. For example, at an \ac{SNR} of 25 dB, increasing $\alpha_p$ from 0.25 to 1 reduces the average $\text{PMD}$ from approximately $10^{-2}$ to $10^{-3}$.The theoretical results are also aligned with the simulation results.

\begin{figure}[t]
    \centering
    \includegraphics[width=0.85\linewidth]{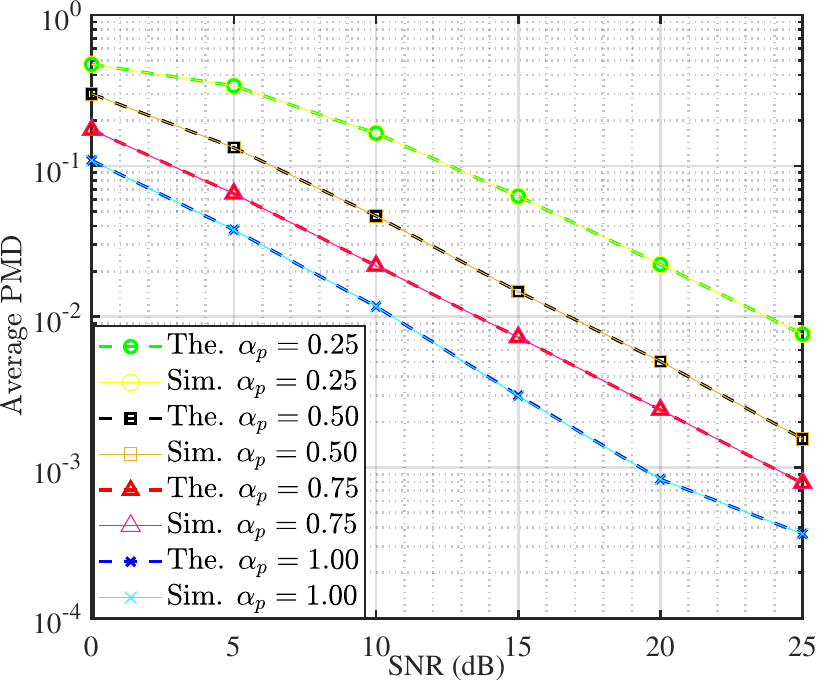}
    \caption{The average PMD of SBC for Fully-Orthogonal scheme with OFSK modulation at $N = 64$.}
    \label{fig:image6}
\end{figure}

In \figurename~\ref{fig:image7}, we provide insights on the average \ac{PMD} performance of the non-coherent detection of \ac{OFSK} modulated signal in Fully-Orthogonal scheme, by examining the \ac{BER} performance for different values of $N$. It is observed that the average \ac{BER} decreases with increasing in $N$. For instance, at $20$ dB SNR value average \ac{PMD} reaches near to $10^{-2}$ with $N=256$. The lowest average \ac{PMD} is achieved at $25$ dB SNR value with $N=512$. Besides, the theoretical average \ac{PMD} curves are aligned with the curves obtained from simulations. This justifies that, increasing $N$ provides greater flexibility in allocating null subcarriers, thereby enhancing system performance. The results indicate a substantial reduction in the average \ac{PMD} as $N$ increases.

In \figurename~\ref{fig:image8}, we analyze the \ac{ROC} of \ac{OFSK} modulation in Fully-Orthogonal scheme for various \ac{SNR} values, for $\alpha_p = 0.25$ and $N = 64$. It is observed that at $\text{PFA}_p = 0.2$ the system achieves a detection probability average $P_{D}$ of $45\%$ at $0$ dB \ac{SNR}. As the \ac{SNR} increases to $5$ dB, the average $ P_{D}$ increases to approximately $70\%$ for the same $\text{PFA}_p$. Further increasing the SNR to $10$ dB improves $P_{D}$ to around $90\%$. These findings suggest that even with a low reflection coefficient in a low-cost \ac{BD}, reliable detection performance can be achieved when the receiver operates under high \ac{SNR} conditions.

\begin{figure}[t]
   \centering
    \includegraphics[width=0.85\linewidth]{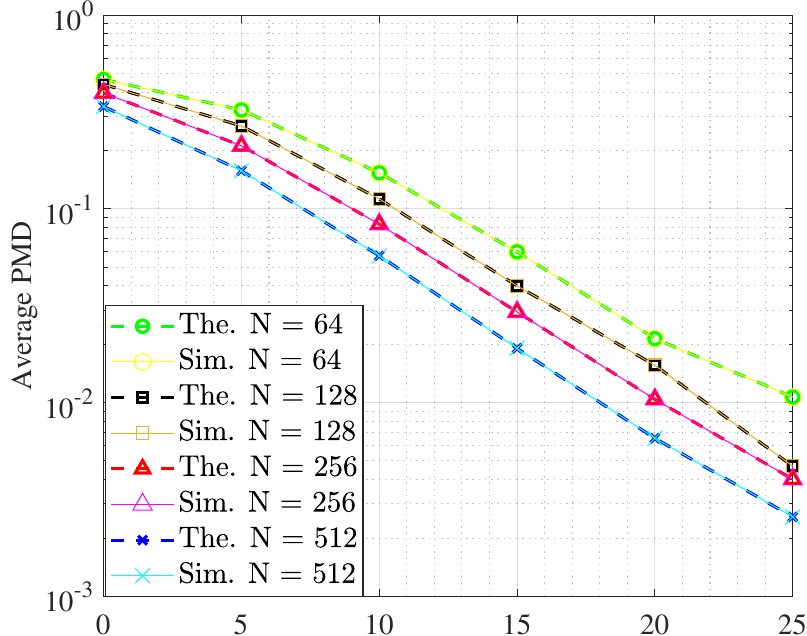}
    \caption{The average PMD of SBC for the Fully-Orthogonal scheme with OFSK modulation at $\alpha_p=0.25$.} 
    \label{fig:image7}
\end{figure}

\begin{figure}[t]  
    \centering
    \includegraphics[width=0.85\linewidth]{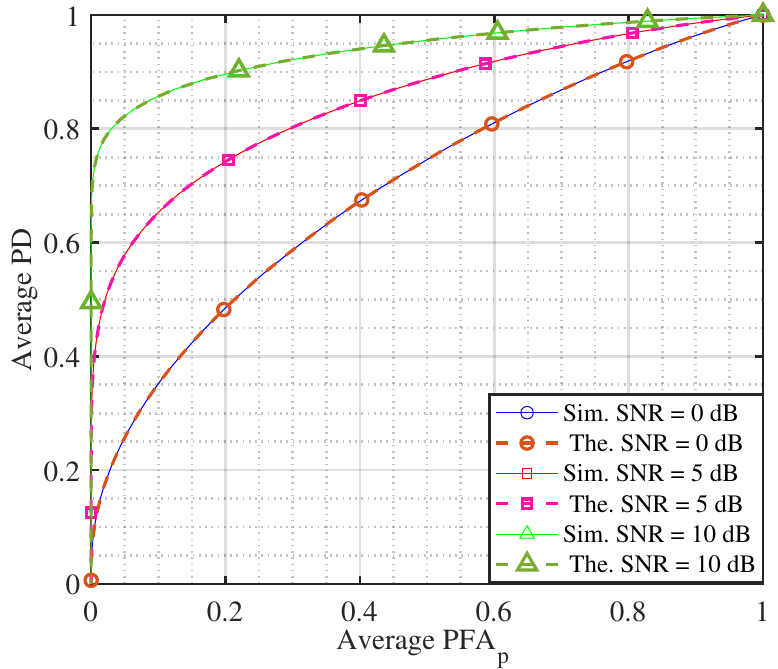} 
    \caption{The ROC of SBC for the Fully-Orthogonal scheme with OFSK modulation at $N = 64$ and $\alpha_p=0.25$.}
    \label{fig:image8}
\end{figure}

In \figurename~\ref{fig:image9}, we analyze the performance of non-coherent detection for the Semi-Orthogonal scheme with \ac{OFSK} modulation by examining the average $\text{PMD}$ for different values of $\alpha_p$. The results indicate that the average \ac{PMD} in the Semi-Orthogonal scheme with \ac{OFSK} is higher compared to the Fully-Orthogonal \ac{OFSK}. For instance, at SNR of 15 dB and $\alpha_p = 0.25$, the average \ac{PMD}  is approximately $10^{-1}$, where after applying \ac{SIC}, a significantly drop to below $10^{-2}$ is observed. This signifies that the effect of SIC could recover the detection performance to the ideal values where the significant decrease in interference is indicated by a decrease in the average $\text{PMD}$.

\begin{figure}[t]  
    \centering
   \includegraphics[width=0.85\linewidth]{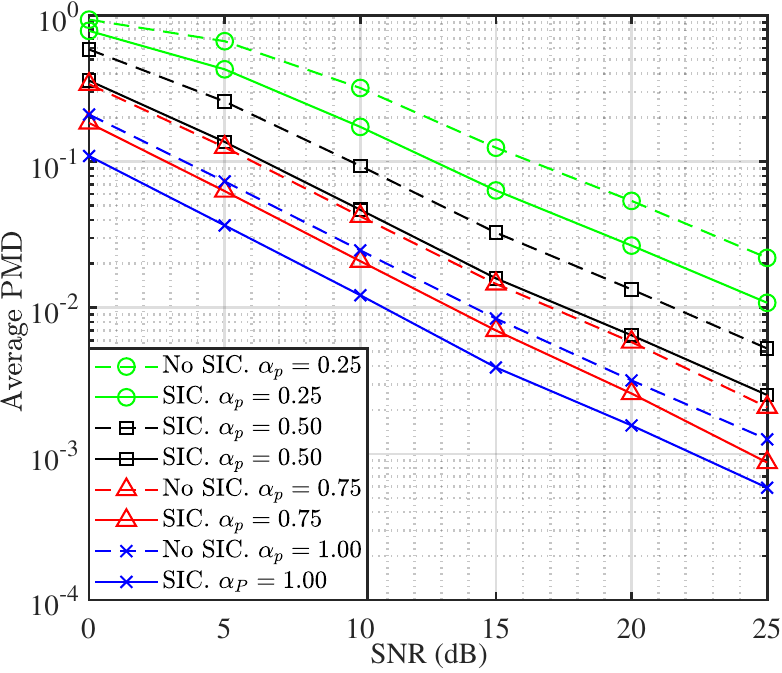} 
    \caption{The average PMD of SBC for the Semi-Orthogonal scheme with OFSK modulation at $N = 64$. }
    \label{fig:image9}
\end{figure}


In \figurename~\ref{fig:image10}, we present the average \ac{BER} performance of non-coherent detection of the \ac{MFSK} modulated signal in the Fully-Orthogonal scheme. The results indicate a consistent reduction in \ac{BER} with increasing \ac{SNR} across all values of $\alpha_p$. It is observed that higher values of $\alpha_p$ lead to significantly improved \ac{BER} performance due to the enhanced backscattered signal power. For instance, at \ac{SNR} of 15 dB, increasing $\alpha_p$ from 0.25 to 1 reduces the \ac{BER} from approximately $10^{-1}$ to $10^{-2}$. Similarly, at an SNR of 25 dB, the \ac{BER} for $\alpha_p = 1$ reaches the lowest value of approximately $10^{-3}$. The results also demonstrate that theoretical predictions align with simulation, validating the analytical model and the performance of the proposed scheme.

\begin{figure}[t]  
    \centering
    \includegraphics[width=0.85\linewidth]{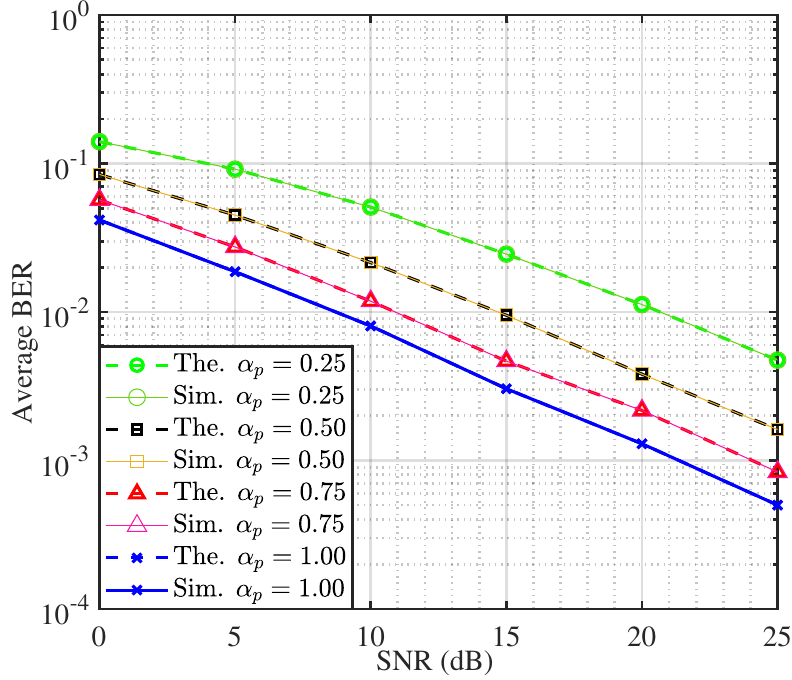} 
    \caption{The average BER of SBC for the Fully-Orthogonal scheme  with MFSK modulation at $N = 64$.}
    \label{fig:image10}
\end{figure}

In \figurename~\ref{fig:image11} we analyze the average \ac{BER} of the Fully-Orthogonal \ac{MFSK} modulated signal under the influence of different values of $N \in \{64, 128, 256, 512\}$. The results indicate that increasing $N$ leads to a significant decrease in \ac{BER}. For instance, at an \ac{SNR} of 25 dB, the average \ac{BER} decreases from above $10^{-2}$ to below $10^{-2}$ as $N$ increases from $128$ to $512$. Indicating that increasing $N$ provides more dedicated resources for the \ac{SR-BC} signal which in turn enhances the average \ac{BER} performance.

\begin{figure}[t]  
    \centering
   \includegraphics[width=0.9\linewidth]{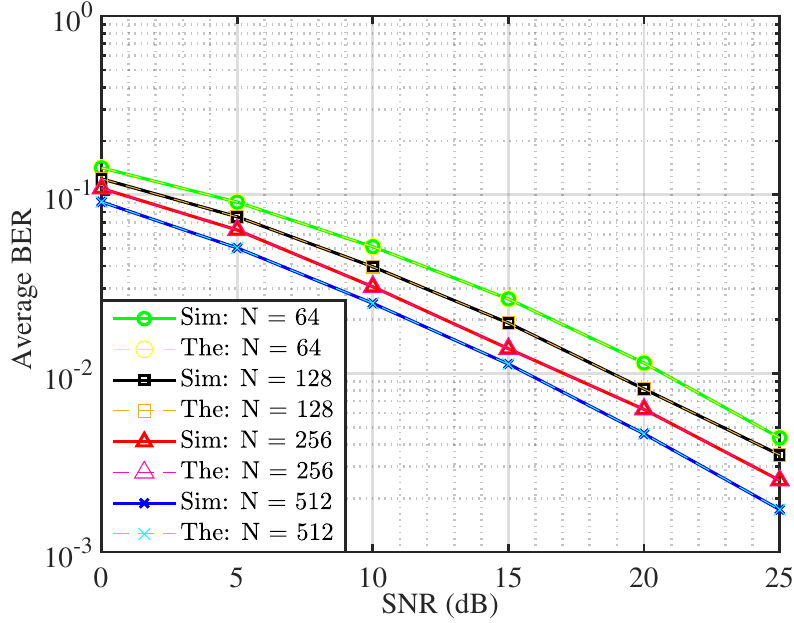} 
   \caption{The average BER of SBC for the Fully-Orthogonal scheme with MFSK modulation at  $\alpha_p=0.25$.}
    \label{fig:image11}
\end{figure}

\figurename~\ref{fig:image12} illustrates the average \ac{BER} of the detector for the \ac{MFSK} modulated signal in the Semi-Orthogonal scheme, considering different values of $\alpha_p$. Allocating fewer null subcarriers causes the BD signal to experience strong \ac{DLI} at non-null designated subcarriers. However, after applying \ac{SIC}, the interference is significantly reduced, resulting in improved average \ac{BER} performance of \ac{SR-BC}, which converges to that of the  Fully-Orthogonal scheme. The average \ac{BER} reduction indicates the effectiveness of the applied \ac{SIC} to mitigate the partial interference and counteract the effects of the cascade channel.

\begin{figure}[t]  
    \centering
   \includegraphics[width=0.85\linewidth]{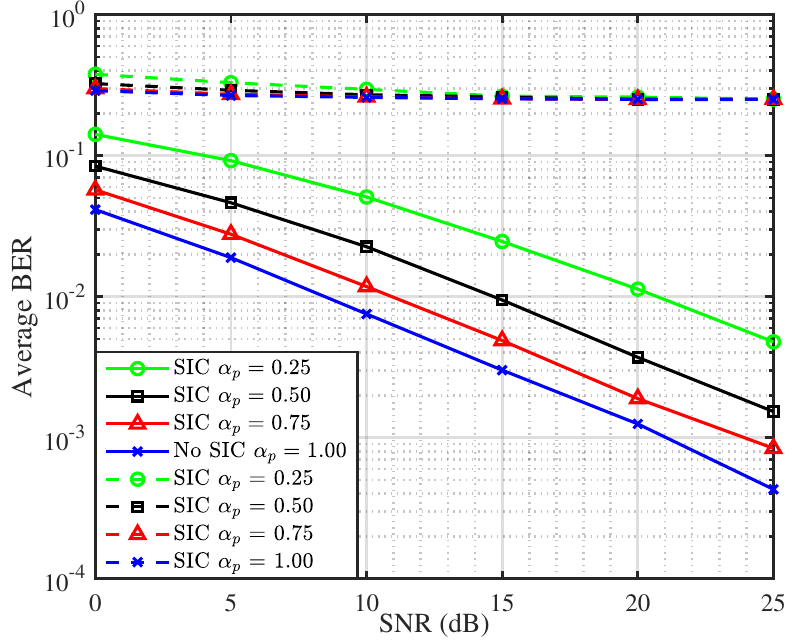} 
    \caption{The average BER of SBC for the Semi-Orthogonal scheme with MFSK modulation, for $N = 64$ with and without SIC.}
    \label{fig:image12}
\end{figure}

In \figurename~\ref{fig:image13}, we analyze the \ac{BER} performance of the non-coherent detection of the \ac{MFSK} modulated signal in the Fully-Orthogonal scheme under \ac{CFO} impairments, with \ac{CFO} $\in  \{0, 0.01, 0.03, 0.05\}$ for $\alpha_p = 0.25$. As it is seen, the \ac{BER} increases as the \ac{CFO} increases. In particular, higher \ac{CFO} values cause high miss detection of the detector which leads to significant performance degradation, expressed by high \ac{BER}. However, implementing \ac{CFO} estimation and compensation effectively restores system performance, reducing the \ac{BER} to the ideal values.

\begin{figure}[t]  
    \centering
   \includegraphics[width=0.90\linewidth]{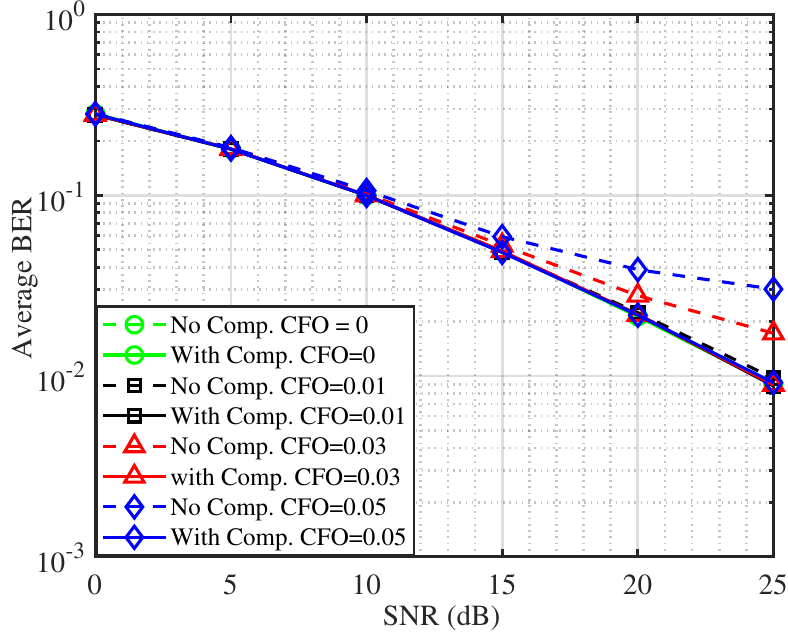} 
    \caption{The average BER SBC for the Fully-Orthogonal scheme with OFSK modulation with and without CFO compensation at $\alpha_p = 0.25$.}
    \label{fig:image13}
\end{figure}

In \figurename~\ref{fig:image14}, we analyze the \ac{BER} performance of a non-coherent detector of \ac{MFSK} modulated signal in the Fully-Orthogonal scheme for the normalized value of $\text{\ac{CFO}} = 0.5$ for various values of $\alpha_p \in \{0.25, 0.5, 0.75, 1\}$. The results show that the errors due to \ac{CFO} impairments increase the \ac{BER} with a decrease in  $\alpha_p$. The lower $\alpha_p$ demonstrates the higher the average \ac{BER} under  \ac{CFO} than the higher $\alpha_p$. However, the effective implementation of \ac{CFO} estimation and compensation successfully restores system performance, bringing the \ac{BER} back to its optimal levels. For instance, at an \ac{SNR} of 25 dB, after \ac{CFO} estimation and compensation the average \ac{BER} for $\alpha_p = 1$ is reduced from the average \ac{BER} value of $0.3$ to $10^{-3}$.

\begin{figure}[t]  
    \centering
   \includegraphics[width=0.85\linewidth]{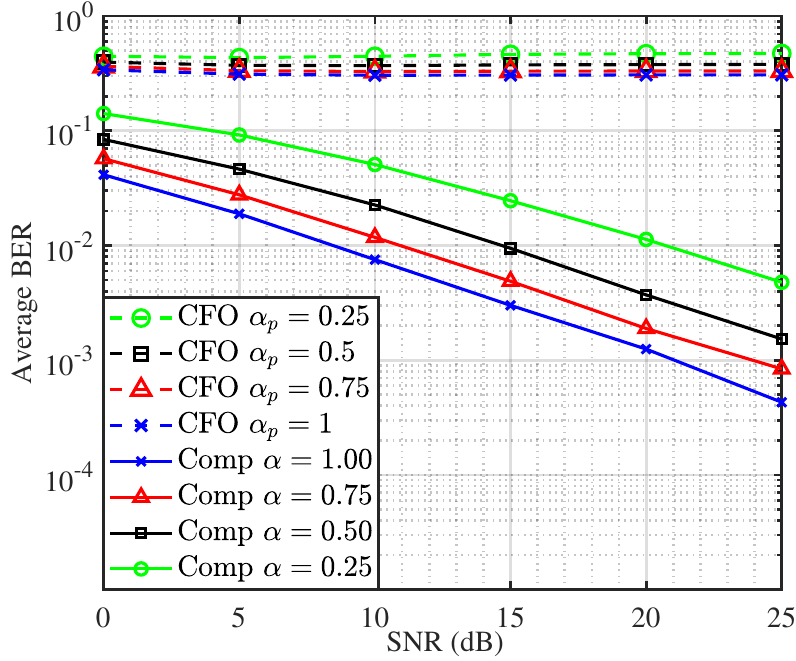} 
    \caption{The average BER of SBC for the Fully-Orthogonal scheme with MFSK modulation at $N=64$ with $\text{CFO}=0.05$ and compensation.   }
    \label{fig:image14}
\end{figure}

In \figurename~\ref{fig:image15}, we compare the system performance \ac{BER} performance of \ac{OFSK} and \ac{MFSK} modulation techniques with the SA scheme for multiple \ac{IoT} devices, as presented in \cite{ref27}. The simulation results demonstrate that the proposed schemes outperform the SA approach. This performance enhancement is attributed to the fact that, in the proposed modulation techniques, the \ac{SR-BC} signal is transmitted over at least a given number of dedicated RF resources which are the null subcarriers, effectively reducing interference from other transmission links. Consequently, the overall system performance is significantly improved. On the other hand, the SA scheme suffers from high interference due to its reliance solely on \ac{SIC} while sharing the same RF time-frequency resources with primary signal transmission. As a result, it requires a higher \ac{SNR} for effective transmission, as illustrated in \figurename~\ref{fig:image15}. This contradicts the ideal low-power design principles, making SA-based approaches less suitable for efficient low-power communication compared to the proposed schemes.

\begin{figure}[t]
    \centering
    \includegraphics[width=0.95\linewidth]{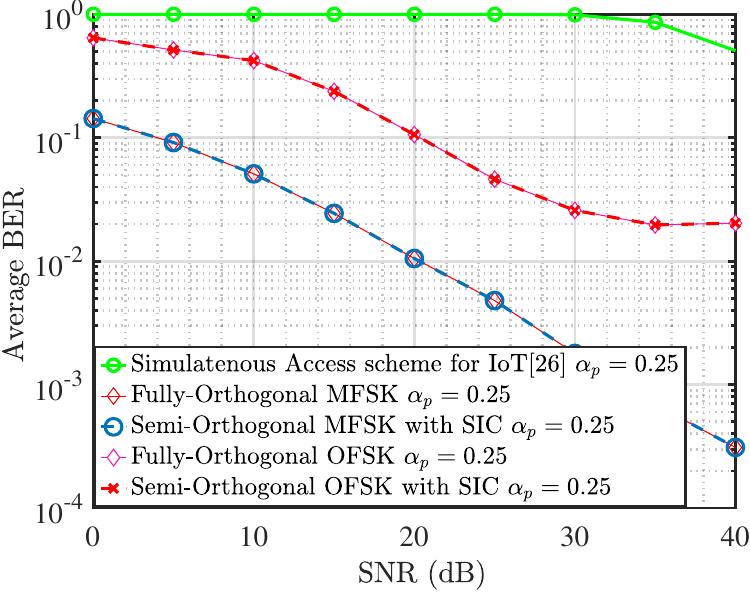}
    \caption{A comparison of the proposed schemes with the baseline approach \cite{ref26}.}
    \label{fig:image15}
\end{figure}

\subsection{Simulation Results for Sum-Rate}

In Fig.~\ref{fig:image16}, we evaluate the total sum-rate performance of various modulation techniques under varying numbers of \acp{BD}. The results show that increasing the reuse of the subcarrier leads to higher total system throughput. In this context, Semi-Orthogonal schemes outperform their Fully-Orthogonal counterparts. Specifically, the Semi-Orthogonal scheme with \ac{OFSK} achieves the highest sum-rate, reaching approximately $0.65$ Mbps for $8$ \acp{BD}, followed by the Semi-Orthogonal scheme with \ac{MFSK} at around $0.55$ Mbps, reflecting efficient use of spectral resources. In contrast, Fully-Orthogonal schemes exhibit significantly lower sum-rates, with  \ac{OFSK} achieving approximately $0.1$ Mbps and \ac{MFSK} reaching only 0.08 Mbps for the same number of \acp{BD}.

\begin{figure}[t]  
    \centering
   \includegraphics[width=0.95\linewidth]{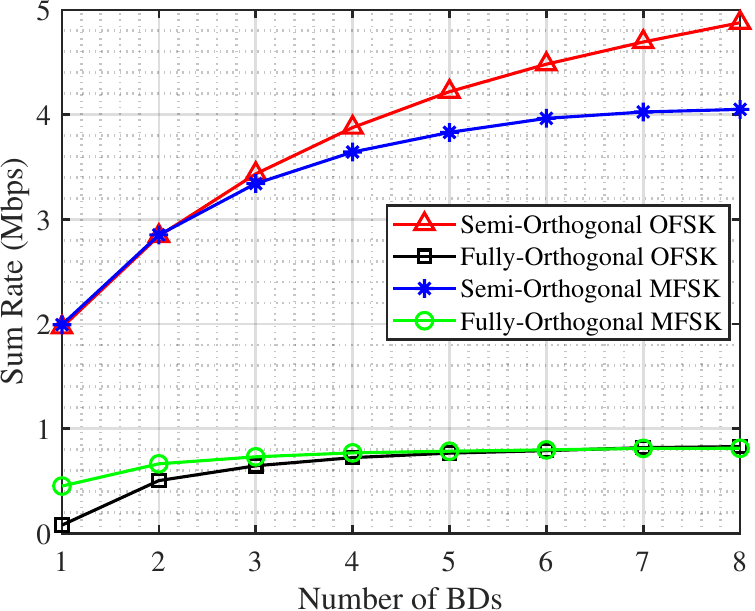} 
    \caption{The total sum-rate of OFDM-based \ac{SR} system at SNR = 20 dB and $\alpha_p = 0.25$ for the four schemes.}
    \label{fig:image16}
\end{figure}

\figurename~\ref{fig:image17} illustrates the total sum-rate performance of the Semi-Orthogonal and  Fully-Orthogonal schemes with both OFSK and MFSK modulations as a function of the number of \acp{BD} at \ac{SNR} of 20 dB for different values of $\alpha_p$. The results indicate that the total sum-rate increases as the number of \acp{BD} increases, demonstrating improved spectral efficiency with additional backscatter links. However, the total sum-rate performance varies significantly between the Semi-Orthogonal and Fully-Orthogonal schemes. In the Semi-Orthogonal scheme, the total sum-rate is higher compared to that in the fully orthogonal case, because of the efficient reuse of spectral resources. This is particularly evident for large values of $\alpha_p$, where the backscatter signals are stronger, leading to improved sum-rate performance. The performance gap between the Semi-Orthogonal and the Fully-Orthogonal schemes is more pronounced when $\alpha_p = 1$. For example, in the presence of $3$ \acp{BD} the Semi-Orthogonal with \ac{OFSK} schemes outperform others by achieving the sum-rate close to $7$ Mbps. Overall, the results highlight the benefits of the Semi-Orthogonal multiple access approach in maximizing the sum-rate while maintaining efficient bandwidth utilization, particularly when stronger \ac{SR-BC} signals are available.

\begin{figure}[t]  
    \centering
   \includegraphics[width=0.94\linewidth]{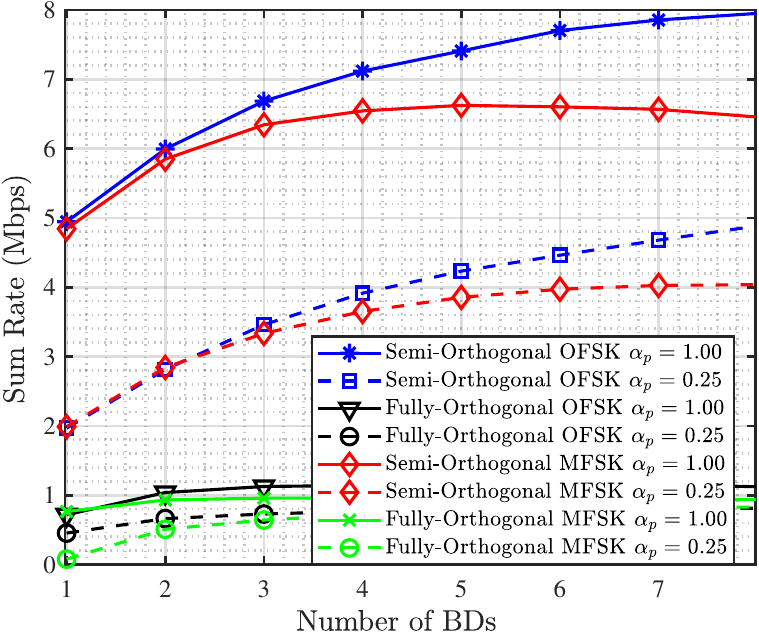} 
    \caption{The Total sum-rate of OFDM-based SR system at SNR = 20 dB.}
    \label{fig:image17}
\end{figure}

In Fig.~\ref{fig:image18}, we analyze the total sum-rate performance of multiple \acp{BD} under varying \ac{SNR}. The results demonstrate that an increase in \ac{SNR} level improves the total sum-rate consistently. The Semi-Orthogonal scheme with \ac{OFSK} scheme achieves the highest total sum-rate, reaching approximately $6.5$ Mbps at 25 dB, attributed to its optimal utilization of subcarriers for signal transmission.  In contrast, the Fully-Orthogonal scheme with \ac{MFSK} shows the lowest sum-rate due to its orthogonal subcarrier allocation, which is less spectrally efficient.

\begin{figure}[t]  
    \centering
   \includegraphics[width=0.95\linewidth]{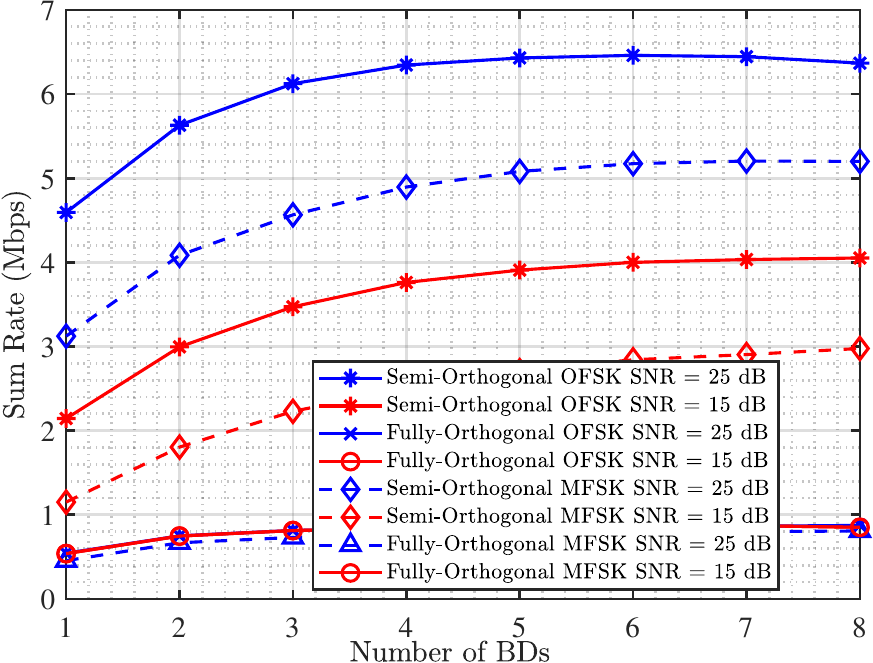} 
    \caption{The total sum-rate of OFDM-based SR system at $\alpha = 0.25$ for various values of SNR (dB).}
    \label{fig:image18}
\end{figure}

\begin{figure}[t]  
    \centering
   \includegraphics[width=0.90\linewidth]{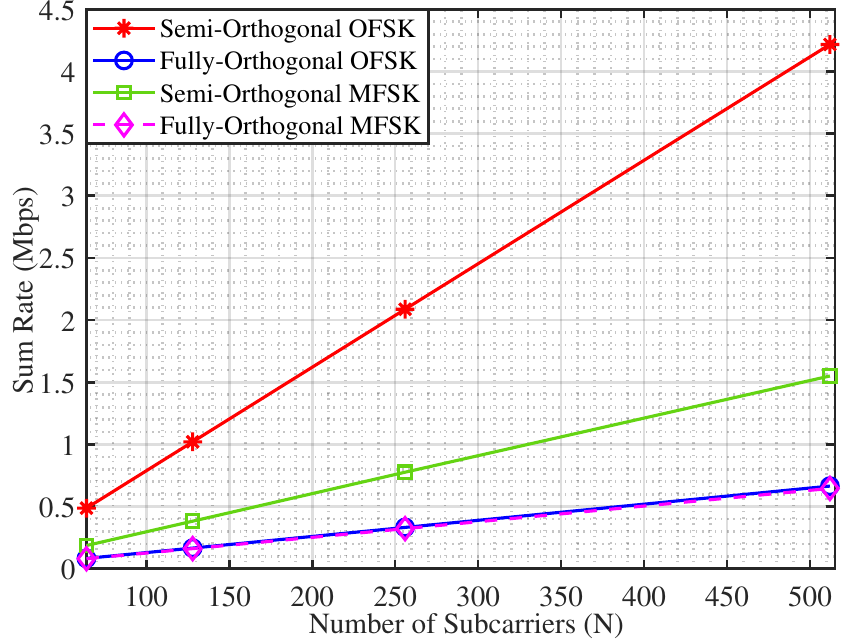} 
    \caption{The sum-rate of SR at $\text{SNR} = 10$ dB and $\alpha_p =0.25$.}
    \label{fig:image19}
\end{figure}

In Fig.~\ref{fig:image19}, we analyze, the total sum-rate performance is analyzed as a function $N$. The results indicate that as $N$ increases, the total sum-rate improves across all schemes. The Semi-Orthogonal scheme with OFSK achieves the highest sum-rate, reaching approximately $4.5$ Mbps at $N = 512$, highlighting its superior spectral efficiency due to effective subcarrier utilization. In comparison, \ac{MFSK} modulation techniques yield substantially lower sum-rates, with the Semi-Orthogonal scheme with \ac{MFSK} scheme reaching about $1.5$ Mbps at the same value of $N$. Moreover, the Semi-Orthogonal scheme consistently outperforms the Fully-Orthogonal scheme with both \ac{OFSK} and \ac{MFSK}, emphasizing the advantage of reduced null subcarrier placement for maximizing overall system performance.

\section{Concluding Remarks}
This study provides a framework for interference mitigation in multi-\ac{BD} \ac{OFDM}-based \ac{SR} systems. The primary \ac{OFDM} signal is designed based on Fully-Orthogonal and Semi-Orthogonal schemes for multiple access of \acp{BD} and the primary user. In these schemes, null and data subcarriers are placed strategically within the band, where null subcarriers serve as dedicated frequency resources for \ac{SR-BC}. Additionally, two modulation techniques including \ac{OFSK} and \ac{MFSK} are introduced at the \acp{BD}. The Fully-Orthogonal schemes demonstrate superior average \ac{BER} performance compared to Semi-Orthogonal schemes. Specifically, the Fully-Orthogonal scheme with \ac{MFSK} modulation achieves the best average \ac{BER} performance, albeit with a substantial loss in spectral efficiency. This loss of spectral efficiency is overcome in the Semi-Orthogonal schemes. Among these schemes, the Semi-Orthogonal scheme with \ac{OFSK} modulation achieves the highest spectral efficiency but with a trade-off in reduced \ac{BER} performance. To mitigate this loss in \ac{BER}, \ac{SIC} is applied within the Semi-Orthogonal schemes. We also analyze the performance of our proposed schemes under the influence of CFO and present a compensation method to mitigate its effects. The \ac{BER} performance of the proposed detector is analytically derived and validated through simulations. The results demonstrate that increasing the \ac{OFDM} symbol size or the reflection coefficient significantly improves both average \ac{BER} performance and the total sum-rate. The proposed approaches can be seamlessly integrated into existing network technologies such as Wi-Fi, LTE, and 5G, enabling a wide range of low-power and ambient \ac{SR-BC} applications. However, achieving high gains in both spectral efficiency and \ac{BER} performance remains a challenge. While the Semi-Orthogonal schemes with \ac{SIC} show enhanced spectrum utilization, the inherently weak nature of \ac{SR-BC} signals still poses limitations. Future research may target advanced interference suppression techniques for the Semi-Orthogonal scheme and optimal resource allocation strategies for achieving high spectral efficiency in the Fully-Orthogonal scheme.

\appendices
\section{The Proof of Lemma~\ref{lemma:1}}
\label{prof:lemma:pdfOFSK}
\textbf{Proof.} The received signal is a function of the product of two random variables: $|h_\mathrm{b}|$, a Rayleigh random variable, and the other channel $|h_{\mathrm{f}p}|^2$, an exponential random variable. Let
\begin{equation}
    u_p = \sum_{m=1}^{N_\mathrm{b}} |h_{\mathrm{f}p}[m]|^2
\end{equation}
be the sum of exponential random variables, while the characteristic function of an exponential random variable
\begin{equation}
    f(x;\lambda_m) = \lambda_m e^{-\lambda_m x}
\end{equation}
is given by its Fourier transform:
\begin{equation}
    \Phi_m(t) = \int e^{(it - \lambda_m)z} dx = \frac{1}{1 + it\lambda_m^{-1}}~,
\end{equation}
where $\Phi_m(t)$ is the Fourier transform of $f(x;\lambda_m)$.
To find the distribution of $u$, we need to compute the convolution of $N_\mathrm{b}$ exponential random variables, which is equivalent to the product of their characteristic functions expressed as
\begin{equation}
\begin{split}
    \Phi_{u_p} =& \int e^{itz} (\lambda_1 e^{-\lambda_1 z} * \lambda_2 e^{-\lambda_2 z} * \cdots * \lambda_{N_\mathrm{b}} e^{-\lambda_{N_\mathrm{b}} z}) dz~, \\ =& \prod_{m=1}^{N_\mathrm{b}} \frac{1}{1 + it\lambda_m}~.
    \end{split}
\end{equation}
Now, the \ac{PDF} of $u$ can be found by taking the inverse Fourier transform of $\Phi_u(t)$ as
\begin{equation}
    f_{u_p}(u_p) = \int_{-\infty}^{\infty} \frac{1}{2\pi j t} \prod_{m=1}^{N_\mathrm{b}} \frac{1}{1 + it\lambda_m^{-1}} dt~.
\end{equation}

\section{The Proof of Theorem}
\label{proof:theorem:pdfmfsk}
\textbf{Proof.} The received signal is a function of the product of two random variables: $|h{_\mathrm{b}}_p|$, a Rayleigh random variable, and the other channel $|h_{\mathrm{f}p}|^2$, an exponential random variable. Let it be expressed as
\begin{equation}
    u_p = \sum_{m=1}^{N_\mathrm{b}} |h_{\mathrm{f}p}[m]|^2~,
\end{equation}
which is the sum of exponential random variables, while the characteristic function of an exponential random variable is given by
\begin{equation}
    f(x;\lambda_m) = \lambda_m e^{-\lambda_m x}~.
\end{equation}
Its Fourier transform is described as
\begin{equation}
    \Phi_m(t) = \int e^{(it - \lambda_m)z} dx = \frac{1}{1 + it\lambda_m^{-1}}~.
\end{equation}

To find the distribution of $u$, we need to compute the convolution of $N_\mathrm{b}$ exponential random variables, which is equivalent to the product of their characteristic functions and expressed as 
\begin{equation}
\begin{split}
    \Phi_{u_p} =& \int e^{itz} (\lambda_1 e^{-\lambda_1 z} \times \lambda_2 e^{-\lambda_2 z} \times \cdots \times \lambda_{N_\mathrm{b}} e^{-\lambda_{N_\mathrm{b}} z}) dz \\ =& \prod_{m=1}^{N_\mathrm{b}} \frac{1}{1 + it\lambda_m}~.
    \end{split}
\end{equation}
Now, the PDF of $u$ can be found by taking the inverse Fourier transform of $\Phi_u(t)$ as
\begin{equation}
    f_{u_p}(u_p) = \int_{-\infty}^{\infty} \frac{1}{2\pi j t} \prod_{m=1}^{N_\mathrm{b}} \frac{1}{1 + it\lambda_m^{-1}} dt~.
\end{equation}

Hence, we obtain the CDF of ${R_p^0} - {R_p^1}$ as
\begin{equation}
    F_{{R_p^0} - {R_p^1}}(\eta) = \int_0^{\infty} F_{{R_p^0} - {R_p^1}}(\eta; v) f(v_p) dv_p~,
\end{equation}
where
\begin{equation}
    \int_0^{\infty} \frac{1}{2} - \int_0^{\infty} \int_{-\infty}^{\infty} \frac{\Phi_{{R_p^0}|v_p}(t) \Phi_{{R_p^1}|v_p}^*(t)}{2\pi j t} e^{-jtz} dt f(v_p) dv_p~.
\end{equation}



\end{document}

%% file: Acronyms.tex
\DeclareAcronym{OOK}{
short = OOK,
long = on-off keying
}

\DeclareAcronym{OFSK}{
short = OFSK,
long = on-off frequency shift keying
}

\DeclareAcronym{OFDM}{
short = OFDM,
long = orthogonal frequency division multiplexing
}

\DeclareAcronym{BC}{
short = BC,
long = backscatter communication
}

\DeclareAcronym{RFID}{
short = RFID,
long = radio frequency identification 
}

\DeclareAcronym{BD}{
short = BD,
long = backscatter device 
}

\DeclareAcronym{SR}{
short = SR,
long = symbiotic radio 
}

\DeclareAcronym{BS}{
short = BS,
long = base station
}

\DeclareAcronym{SR-BC}{
short = SBC,
long = symbiotic backscatter communication
}

\DeclareAcronym{CP}{
short = CP,
long = cyclic prefix
}

\DeclareAcronym{FSK}{
short = FSK,
long = frequency shift keying
}

\DeclareAcronym{MFSK}{
short = MFSK,
long = multiple frequency shift keying
}

\DeclareAcronym{DLI}{
short = DLI,
long = direct-link interference
}

\DeclareAcronym{Rx}{
short = Rx,
long = legacy receiver
}

\DeclareAcronym{SIC}{
short = SIC,
long = successive interference cancellation
}

\DeclareAcronym{ISI}{
short = ISI,
long = inter-symbol interference
}

\DeclareAcronym{MCU}{
short = MCU,
long = microcontroller unit
}

\DeclareAcronym{STO}{
short = STO,
long  = symbol time offset
}

\DeclareAcronym{CFO}{
short = CFO,
long  = carrier frequency offset
}

\DeclareAcronym{ICI}{
short = ICI,
long  = inter-carrier interference
}

\DeclareAcronym{3GPP}{
short = 3GPP,
long  = 3rd generation partnership project
}

\DeclareAcronym{IoT}{
short = IoT,
long  = Internet of Things
}

\DeclareAcronym{5G-NR}{
short = 5G-NR,
long  = 5G New Radio
}

\DeclareAcronym{PMD}{
short = PMD,
long  = probability of missed detection
}

\DeclareAcronym{PD}{
short = PD,
long  = probability of Detection
}

\DeclareAcronym{PFA}{
short = PFA,
long  = probability of false alarm
}

\DeclareAcronym{SNR}{
short = SNR,
long  = signal-to-noise ratio
}

\DeclareAcronym{ROC}{
short = ROC,
long  = receiver operating character
}

\DeclareAcronym{BER}{
short = BER,
long  = bit error rate
}

\DeclareAcronym{SINR}{
short = SINR,
long  = signal-to-noise interference ratio
}

\DeclareAcronym{AWGN}{
short = AWGN,
long  = additive white Gaussian noise
}

\DeclareAcronym{RF}{
short = RF,
long  = radio-frequency
}

\DeclareAcronym{Nb}{
short = Nb,
long  = number of null subcarriers dedicated for SR -BC
}
\DeclareAcronym{Nd}{
short = Nd,
long  = number of subcarriers for primary data transmission
}

\DeclareAcronym{Ncp}{
short = Ncp,
long  = CP length
}

\DeclareAcronym{DFT}{
short = DFT,
long  = discrete Fourier Transform
}

\DeclareAcronym{PDF}{
short = PDF,
long  = probability density function
}

\DeclareAcronym{CDF}{
short = CDF,
long  = cumulative distribution function
}

\DeclareAcronym{PSS}{
short = PSS,
long  = primary synchronization signal
}

\DeclareAcronym{SSS}{
short = SSS,
long  = secondary synchronization signal
}

\DeclareAcronym{IBDI}{
short = IBDI,
long  = inter-backscatter device interference
}

\DeclareAcronym{SO}{
short = SO,
long  = semi orthogonal
}

\DeclareAcronym{FO}{
short = FO,
long  = fully orthogonal
}